\documentclass{iopart}

\usepackage{graphicx}
\usepackage{iopams}
\usepackage{amsopn}
\usepackage{amsthm}
\usepackage{algorithmic}
\usepackage{algorithm}
\usepackage{enumerate}
\usepackage[colorlinks=true, citecolor=blue, filecolor=black,linkcolor=red, urlcolor=blue, hypertexnames=false]{hyperref}

\usepackage{todonotes}

\graphicspath{{figures/}}

\newcommand{\eqref}[1]{(\ref{#1})}
\newcommand{\ppgnh}{ppGNH }

\newcommand{\param}{x}
\newcommand{\data}{y}
\newcommand{\error}{e}
\newcommand{\forward}{G}
\newcommand{\linear}{J}
\newcommand{\hessian}{H}
\newcommand{\paramspace}{\mathbb{X}}
\newcommand{\dataspace}{\mathbb{Y}}

\newcommand{\misfit}{\eta}

\newcommand{\rank}{\mathrm{rank}}
\newcommand{\myspan}{\mathrm{span}}

\newcommand{\prmean}{\mu_{\rm pr}}
\newcommand{\prcov}{\Gamma_{\rm pr}}
\newcommand{\pomean}{\mu_{\rm pos}}
\newcommand{\pocov}{\Gamma_{\rm pos}}
\newcommand{\obscov}{\Gamma_{\rm obs}}

\newcommand{\chol}{L}
\newcommand{\basis}{U}
\newcommand{\basisf}{u}
\newcommand{\basisfv}{v}
\newcommand{\gbasis}{\Phi}
\newcommand{\gbasisv}{\Psi}

\newcommand{\gbasisfv}{\psi}
\newcommand{\projlis}{\Pi_r}
\newcommand{\projcs}{\Pi_\perp}

\newcommand{\normal}{\mathcal{N}}
\newcommand{\real}{\mathbb{R}}

\newcommand{\expect}{\mathbb{E}}
\newcommand{\expectpos}{\expect_{\pi}}
\newcommand{\expectapprox}{\expect_{\tilde{\pi}}}

\newcommand{\MC}{Q_N}
\newcommand{\MCRB}{\tilde{Q}_N}

\newcommand{\CR}{\mathcal{R}}

\DeclareMathOperator{\range}{range}

\newtheorem{theorem}{Theorem}%

\newtheorem{lemma}{Lemma}
\newtheorem{problem}{Problem}%
\newtheorem{algo}{Algorithm}

\newtheorem{example}{Example}
\newtheorem{definition}{Definition}

\begin{document}

\title{Likelihood-informed dimension reduction for nonlinear inverse problems}
\author{T Cui$^1$, J Martin$^2$, Y M Marzouk$^1$, A Solonen$^{1,3}$, A Spantini$^1$}
\address{$^1$Department of Aeronautics and Astronautics, Massachusetts Institute of Technology, 77 Massachusetts Ave, Cambridge, MA 02139 USA}
\address{$^2$Institute for Computational Engineering and Sciences, The University of Texas at Austin, 201 East 24th St, Austin, TX 78712 USA}
\address{$^3$Lappeenranta University of Technology, Department of Mathematics and Physics, 53851 Lappeenranta, Finland}
\ead{tcui@mit.edu, jmartin@ices.utexas.edu, ymarz@mit.edu, solonen@mit.edu, spantini@mit.edu}

\begin{abstract} 

  The intrinsic dimensionality of an inverse problem is affected by
  prior information, the accuracy and number of observations, and the
  smoothing properties of the forward operator. From a Bayesian
  perspective, changes from the prior to the posterior may, in many
  problems, be confined to a relatively low-dimensional subspace of
  the parameter space. We present a dimension reduction approach that
  defines and identifies such a subspace, called the
  ``likelihood-informed subspace'' (LIS), by characterizing the
  relative influences of the prior and the likelihood over the support
  of the posterior distribution. This identification enables new and
  more efficient computational methods for Bayesian inference with
  nonlinear forward models and Gaussian priors. In particular, we
  approximate the posterior distribution as the product of a
  lower-dimensional posterior defined on the LIS and the prior
  distribution marginalized onto the complementary subspace. Markov
  chain Monte Carlo sampling can then proceed in lower dimensions,
  with significant gains in computational efficiency. We also
  introduce a Rao-Blackwellization strategy that de-randomizes Monte
  Carlo estimates of posterior expectations for additional variance
  reduction. We demonstrate the efficiency of our methods using two
  numerical examples: inference of permeability in a groundwater
  system governed by an elliptic PDE, and an atmospheric remote
  sensing problem based on Global Ozone Monitoring System (GOMOS)
  observations.
  

\end{abstract}

\maketitle

\section{Introduction}
\label{sec:intro}

Inverse problems arise from indirect observations of parameters of
interest. The Bayesian approach to inverse problems formalizes the
characterization of these parameters through exploration of the
\textit{posterior distribution} of parameters conditioned on data
\cite{Tarantola_2004, Kaipio_2005, Stuart_2010}. Computing
expectations with respect to the posterior distribution yields not
only point estimates of the parameters (e.g., the posterior mean), but
a complete description of their uncertainty via the posterior
covariance and higher moments, marginal distributions, quantiles, or
event probabilities. Uncertainty in parameter-dependent predictions
can also be quantified by integrating over the posterior distribution.

The parameter of interest in inverse problems is often a function of
space or time, and hence an element of an infinite-dimensional
function space \cite{Stuart_2010}. In practice, the parameter field
must be discretized, and the resulting inference problem acquires a
high but finite dimension. The computation of posterior expectations
then proceeds via posterior sampling, most commonly using Markov chain
Monte Carlo (MCMC) methods \cite{MCMC_practice, Liu_2001,
  Handbook_MCMC}. The computational cost and efficiency of an MCMC
scheme can be strongly affected by the parameter dimension,
however. The convergence rates of standard MCMC algorithms usually
degrade with parameter dimension \cite{RGG_1997, Roberts_1998,
  Roberts_2001, Mattingly_2012, PST_2012}; one manifestation of this
degradation is an increase in the mixing time of
the chain, which in turn leads to higher variance in posterior
estimates. Some recent MCMC algorithms, formally derived in the
infinite-dimensional setting \cite{BRSV_2008, CRSW_2012}, do not share
this scaling problem. Yet even in this setting, we will argue that
significant variance reduction can be achieved through explicit
dimension reduction and through de-randomization of posterior
estimates, explained below.

This paper proposes a method for \textit{dimension reduction} in
Bayesian inverse problems. We reduce dimension by identifying a
subspace of the parameter space that is \textit{likelihood-informed};
this notion will be precisely defined in a relative sense, i.e.,
relative to the prior. Our focus is on problems with nonlinear forward
operators and Gaussian priors, but builds on low-rank approximations
\cite{Flath_etal_2011} and optimality results \cite{Linear_Redu_2014}
developed for the linear-Gaussian case. Our dimension reduction
strategy will thus reflect the combined impact of prior smoothing, the
limited accuracy or number of observations, and the smoothing
properties of the forward operator.
Identification of the likelihood-informed subspace (LIS) will let us
write an approximate posterior distribution wherein the distribution
on the complement of this subspace is taken to be independent of the
data; in particular, the posterior will be approximated as the product
of a low-dimensional posterior on the LIS and the marginalization of
the prior onto the complement of the LIS.  The key practical benefit
of this approximation will be \textit{variance reduction} in the
evaluation of posterior expectations. First, Markov chain Monte Carlo
sampling can be restricted to coordinates in the likelihood-informed
space, enabling greater sampling efficiency---i.e., more independent
samples in a given number of MCMC steps or a given computational
time. Second, the product form of the approximate posterior will allow
sampling in the complement of the likelihood-informed space to be
avoided altogether, thus producing Rao-Blackwellized or analytically
conditioned estimates of certain posterior expectations.

Dimension reduction for inverse problems has been previously pursued
in several ways. \cite{Marzouk_2009} constructs a low dimensional
representation of the parameters by using the truncated
Karhunen-L\`{o}eve expansion \cite{Karhunen_1947, Loeve_1978} of the
prior distribution. A different approach, combining prior and
likelihood information via low-rank approximations of the
prior-preconditioned Hessian of the log-likelihood, is used in
\cite{Flath_etal_2011} to approximate the posterior covariance in
linear inverse problems. In the nonlinear setting, low-rank
approximations of the prior-preconditioned Hessian are used to
construct proposal distributions in the stochastic Newton MCMC method
\cite{Martin_2012} or to make tractable Gaussian approximations at the
posterior mode in \cite{Petra_stochnewton2013}---either as a Laplace
approximation, as the proposal for an independence MCMC sampler, or as
the fixed preconditioner for a stochastic Newton proposal.  We note
that these schemes bound the tradeoff between evaluating Hessian
information once (at the posterior mode) or with every sample (in
local proposals). In all cases, however, MCMC sampling proceeds in the
full-dimensional space.

The dimension reduction approach explored in this paper, by contrast,
confines sampling to a lower-dimensional space. We extend the
posterior approximation proposed in \cite{Linear_Redu_2014} to the
nonlinear setting by making essentially a low-rank approximation of
the \textit{posterior expectation} of the prior-preconditioned
Hessian, from which we derive a projection operator. This projection
operator then yields the product-form posterior approximation
discussed above, which enables variance reduction through
lower-dimensional MCMC sampling and Rao-Blackwellization of posterior
estimates.

We note that our dimension reduction approach does not depend on the
use of any specific MCMC algorithm, or even on the use of MCMC. The
low-dimensional posterior defined on coordinates of the LIS is
amenable to a range of posterior exploration or integration
approaches. We also note that the present analysis enables the
construction of dimension-independent analogues of existing MCMC
algorithms with essentially no modification.
This is possible because in inverse problems with formally
discretization-invariant posteriors---i.e., problems where the forward
model converges under mesh refinement and the prior distribution
satisfies certain regularity conditions \cite{LSS_2009,
  Stuart_2010}---the LIS can also be discretization invariant.  We
will demonstrate these discretization-invariance properties
numerically.

The rest of this paper is organized as follows.
In Section 2, we briefly review the Bayesian formulation for inverse problems. 
In Section 3, we introduce the likelihood-informed dimension reduction
technique, and present the posterior approximation and
reduced-variance Monte Carlo estimators based on the LIS. We also
present an algorithm for constructing the likelihood-informed
subspace.
In Section 4, we use an elliptic PDE inverse problem to demonstrate
the accuracy and computational efficiency of our posterior estimates
and to explore various properties of the LIS, including its dependence
on the data and its discretization invariance.
In Section 5, we apply our variance reduction technique to an
atmospheric remote sensing problem.
Section 6 offers concluding remarks.

\section{Bayesian formulation for inverse problems}
\label{sec:background}

This section provides a brief overview of the Bayesian framework for the inverse problems as introduced in \cite{Tarantola_2004, Kaipio_2005, Stuart_2010}.
Consider the inverse problem of estimating parameters $\param$ from data $\data$, where
\begin{equation}
\data = \forward(\param) + \error \, .
\end{equation}
Here $\error$ is a random variable representing noise and/or model error, which appears additively, and $\forward$ is a known mapping from the parameters to the observables.
In a Bayesian setting, we model the parameters $\param$ as a random variable and, for simplicity, assume that the range of this random variable is a finite dimensional space $\paramspace \subseteq \real^{n}$. Then the parameter of interest is characterized by its posterior distribution conditioned on a realization of the data, $\data \in \dataspace \subseteq \real^{d}$:
\begin{equation}
\label{eq:post}
\pi(\param|\data) \propto \pi( \data | \param )  \pi_0(\param).
\end{equation}
We assume that all distributions have densities with respect to Lebesgue measure. The posterior probability density function above is the product of two terms: the prior density $\pi_0(\param)$, which models knowledge of the parameters before the data are observed, and the likelihood function $\pi( \data | \param )$, which describes the probability distribution of $\data$ for any value of $\param$.

We assume that the prior distribution is a multivariate Gaussian $\normal(\prmean, \prcov)$, where the covariance matrix $\prcov$ can be also defined by its inverse, $\prcov^{-1}$, commonly referred to as the precision matrix. We model the additive noise with a zero mean Gaussian distribution, i.e., $\error \sim \normal(0, \obscov)$. This lets us define the data-misfit function
\begin{equation}
\misfit(\param) = \frac{1}{2}  \left\| \obscov^{-\frac12} \left( \forward(\param) - \data \right) \right\|^2,
\label{eq:misfit}
\end{equation}
such that the likelihood function is proportional to $\exp \left( -\misfit(\param) \right )$.

\section{Methodology}
\label{sec:methodology}

\subsection{Optimal dimension reduction for linear inverse problems}
Consider a linear forward model, $\forward(\param) = \forward\param$, with a Gaussian likelihood and a Gaussian prior as defined in Section \ref{sec:background}. The resulting posterior is also Gaussian, $\pi(\param \vert \data) = \normal(\pomean, \pocov)$, with mean and covariance given by 
\begin{equation}
\pomean = \pocov^{} \left( \prcov^{-1} \prmean^{} +  \forward^\top \obscov^{-1} \data \right) \ \  \mathrm{and} \ \ \pocov^{} = \left(\hessian^{} + \prcov^{-1}\right)^{-1} ,
\label{eq:pos_cov}
\end{equation}
where $\hessian = \forward ^\top \obscov^{-1} \forward^{}$ is the Hessian of the data-misfit function \eqref{eq:misfit}. Without loss of generality we can assume zero prior mean and a positive definite prior covariance matrix. 

Now consider approximations to the posterior distribution of the form
\begin{equation} \label{eq:approx_post_linear}
\tilde{\pi}(\param \vert \data) \propto \pi\left(\data  \vert  P_r \param\right) \pi_0(\param), 
\end{equation}
where $P_r = P_r^2$ is a rank-$r$ projector and $\pi\left(\data \vert P_r \param\right)$ is an approximation to the original likelihood function $\pi  \left ( \data \vert \param \right )$. Approximations of this form can be computationally advantageous when operations involving the prior (e.g., evaluations or sampling) are less expensive than those involving the likelihood. As described in \cite{Linear_Redu_2014}, they are also the natural form with which to approximate a Bayesian update, particularly in the inverse problem setting with high-dimensional $\param$. In the deterministic case, inverse problems are ill-posed; the data cannot inform all directions in the parameter space. Equivalently, the spectrum of $H$ is compact or decays quickly. Thus one should be able to explicitly project the argument of the likelihood function onto a low-dimensional space without losing much information in the process. The posterior covariance remains full rank, but the update from prior covariance to posterior covariance will be low rank. The challenge, of course, is to find the best projector $P_r$ for any given $r$. The answer will involve balancing the influence of the prior and the likelihood.
In the following theorem, we introduce the optimal projector and  characterize its approximation properties.
\begin{theorem}
\label{minimizer_theorem}
Let $\prcov = \chol \chol^\top$ be a symmetric decomposition of the prior covariance matrix  and let $(\lambda_i, \basisfv_i)$ be the  eigenvalue-eigenvector pairs of the  prior-preconditioned Hessian $\left(\chol^\top \hessian \chol \right)$ such that  $\lambda_i \geq \lambda_{i+1}$. Define the directions $u_i= L v_i$ and $w_i= L^{-\top}v_i$ together with the matrices $U_r=\left[ u_1 , \ldots , u_r \right]$ and $W_r=\left[ w_1 , \ldots , w_r \right]$. Then, the projector $P_r$ given by:
\[
P_r^{} = \basis_r^{} W_r^\top,
\]
yields an approximate posterior density of the form  $\tilde{\pi}(\param \vert \data)=\mathcal{N}\left( \pomean^{(r)}, \pocov^{(r)}\right)$ and  is optimal in the following sense: 
\begin{enumerate}
 \item  $\pocov^{(r)}$ minimizes the F\"{o}rstner  distance  \cite{FB_1999} from the exact posterior covariance over the class of positive definite matrices that can be written as rank $r$ negative semidefinite updates of the prior covariance.
  \item  $\pomean^{(r)} = A^\ast y$ minimizes the Bayes risk\ $\mathbb{E}_{x,y} \left[ \left\|   \mu(y)  - x \right\|_{\Gamma_{\rm pos}^{-1}}^2 \right]$ over the class of all linear transformations of the data $\mu(y) = A y$ with $\rank(A) \leq r$.
\end{enumerate} 
\end{theorem}

\begin{proof}
We refer the reader to \cite{Linear_Redu_2014} for a proof and detailed discussion. 
\end{proof}

The vectors $(u_1,\ldots,u_r)$ span the range of the optimal projector; we call this range the \textit{likelihood-informed subspace} of the linear inverse problem. Note that the $(u_i)$ are eigenvectors of the pencil $(H,\Gamma_{\rm pr}^{-1})$. Hence, the $j$th basis vector $\basisf_j$ maximizes the Rayleigh quotient
\begin{equation}
\CR(\basisf) = \frac{\langle \basisf, \hessian \basisf\rangle}{ \langle \basisf,  \prcov^{-1} \basisf \rangle}
\label{eq:ray_quo_u}
\end{equation}
over the subspace $\paramspace \setminus \myspan\{\basisf_1, \ldots, \basisf_{j-1}\}$. This Rayleigh quotient helps interpret the $(u_i)$ as directions where the data are most ``informative'' relative to the prior. For example, consider a direction $w \in \paramspace$ representing a rough mode in the parameter space. If the prior is smoothing, then the denominator of (\ref{eq:ray_quo_u}) will be large; also, if the forward model output is relatively insensitive to variation in the $w$ direction, the numerator of (\ref{eq:ray_quo_u}) will be small. Thus the Rayleigh quotient will be small and $w$ is not particularly data-informed relative to the prior. Conversely, if $w$ is smooth then the prior variance in this direction may be large and the likelihood may be relatively constraining; this direction is then data-informed. Of course, there are countless intermediate cases, but in general, directions for which (\ref{eq:ray_quo_u}) are large will lie in the range of $U_r$.

Note also that $U_r$ diagonalizes both $H$ and $\prcov^{-1}$. We are particularly interested in the latter property: the modes $(u_i)$ are orthogonal (and can be chosen orthonormal) with respect to the inner product induced by the prior precision matrix. This property will be preserved later in the nonlinear case, and will be important to our posterior sampling schemes. 

\medskip

For nonlinear inverse problems, we seek an approximation of the
posterior distribution in the same form as
\eqref{eq:approx_post_linear}. In particular, the range of the
projector will be determined by blending together local
likelihood-informed subspaces from regions of high posterior
probability. The construction of the approximation will be detailed in
the following section.

\subsection{LIS construction for nonlinear inverse problems}

When the forward model is nonlinear, the Hessian of the data-misfit function varies over the parameter space, and thus the likelihood-informed directions are embedded in some nonlinear manifold.
We aim to construct a global linear subspace to capture the majority of this nonlinear likelihood-informed manifold.

Let the forward model $\forward(\param)$ be first-order differentiable. The linearization
of the forward model at a given parameter value $\param$, $
\linear(\param) \equiv \nabla \forward(\param)$ where $\linear(\param) \in \real^{d \times n}$, provides the local sensitivity of the parameter-to-observable map.
Inspired by the dimension reduction approach for the linear inverse problem, we use the linearized forward model $ \linear(\param)$ to construct the Gauss-Newton approximation of the Hessian of the data-misfit function, $\hessian(\param) = \linear(\param) ^\top \obscov^{-1} \linear(\param)$.
Now consider a local version of the Rayleigh quotient
(\ref{eq:ray_quo_u}), 
$$
\CR(\basisf;  \param) := \frac{\langle \basisf, \hessian(\param) \basisf\rangle}{ \langle \basisf,  \prcov^{-1} \basisf \rangle} .
$$
Introducing the change of variable $v = \chol^{-1} u$, we can equivalently use
\begin{equation}
\widetilde{\CR}(\basisfv; \param) := \frac{\langle \basisfv, (\chol^\top \hessian(\param) \chol) \basisfv\rangle}{ \langle \basisfv,
  \basisfv \rangle} = \CR(\chol \basisfv; \param), 
\label{eq:ray_quo_x}
\end{equation}
to quantify the local impact of the likelihood relative to the prior. 
As in the linear problem, this suggests the following procedure for computing a local LIS given some truncation threshold $\tau_{loc}$:
\begin{problem}[Construction of the local likelihood-informed subspace]
\label{pro:l_lis}
Given the Gauss-Newton Hessian of the data misfit function $\hessian(\param)$ at a given $\param$, find the eigendecompostion of the prior-preconditioned Gauss-Newton Hessian (ppGNH) 
\begin{equation}
\chol^\top \hessian(\param) \chol \basisfv_i = \lambda_i \basisfv_i.
\label{eq:eig}
\end{equation}
Given a truncation threshold $\tau_{loc} > 0$, the local LIS is spanned by $\basis_l = [\basisf_1, \ldots, \basisf_l]$, where $\basisf_i = L \basisfv_i$ corresponds to the $l$ leading eigenvalues such that $\lambda_1 \geq \lambda_2 \geq \ldots \geq \lambda_l \geq \tau_{loc}$.
\end{problem}

For a direction $u$ with $\CR( u; x) = 1$, the local impact of the likelihood and the prior are balanced. Thus, to retain a comprehensive set of likelihood-informed directions,  we typically choose a truncation threshold $\tau_{loc}$ less than 1. 

\medskip

To extend the pointwise criterion \eqref{eq:ray_quo_x} into a global criterion for likelihood-informed directions, we consider the expectation of the Rayleigh quotient over the posterior 
\[
\expectpos\left[\CR(\basisf; \param)\right] = \expectpos\left[ \widetilde{\CR}(\basisfv; \param)\right] = \frac{\langle \basisfv, S \basisfv\rangle}{ \langle \basisfv, \basisfv \rangle},
\]
where $S$ is the expected \ppgnh over the posterior,
\begin{equation}
S = \int_{\paramspace} \chol^\top H(\param) \chol \, \pi(d\param \vert \data).
\label{eq:expect_hessian}
\end{equation}
Then we can naturally construct the global LIS through the eigendecomposition of $S$ as in the linear case.
We consider approximating $S$ using the Monte Carlo estimator
\[
\widehat{ S }_n = \frac{1}{n} \sum_{k = 1}^{n} \chol^\top H(\param^{(k)}) \chol,
\]
where $\param^{(k)} \sim \pi(\param \vert \data)$, $k = 1 \ldots n$,  are posterior samples. 
Since the local Hessian $H(\param^{(k)})$ is usually not explicitly available and is not feasible to store for large-scale problems, we use its prior-preconditioned low-rank approximation as defined in Problem \ref{pro:l_lis}.
Thus the global LIS can be constructed by the following procedure:  
\begin{problem} [Construction of global likelihood-informed subspace]
\label{pro:g_lis}
Suppose we have a set of posterior samples $\mathcal{X} = \{ \param^{(k)} \}$, $k=1 \ldots m$, where for each sample $\param^{(k)}$, the \ppgnh is approximated by the truncated low rank eigendecomposition
\[
\chol^\top  H(\param^{(k)}) \chol \approx \sum_{i = 1}^{l(k)} \lambda_i^{(k)} v_i^{(k)} v_i^{(k)\top},
\]
by solving Problem \ref{pro:l_lis}. 
We have $\lambda_i^{(k)} \geq \tau_{loc}$ for all $k = 1 \ldots m$ and all $i = 1 \ldots l(k)$.
To construct the global LIS, we consider the eigendecompostion of the Monte Carlo estimator of the expected Hessian in (\ref{eq:expect_hessian}), which takes the form
\begin{equation}
\left( \frac{1}{m} \sum_{k = 1}^{m} \sum_{i = 1}^{l(k)} \lambda_i^{(k)} \basisfv_i^{(k)} \basisfv_i^{(k)\top} \right) \gbasisfv_j = \gamma_j \gbasisfv_j.
\label{eq:glis}
\end{equation}
The global LIS has the non-orthogonal basis $\gbasis_r = \chol \gbasisv_r$, where the eigenvectors $\gbasisv_r = [\gbasisfv_1, \ldots, \gbasisfv_r]$ correspond to the $r$ leading eigenvalues of  (\ref{eq:glis}), $\gamma_1 \geq \ldots \geq \gamma_{r} \geq \tau_{ g}$, for some truncation threshold $\tau_{g} > 0$.
Here we choose $\tau_{g}$ to be equal to the threshold $\tau_{ loc}$  in Problem \ref{pro:l_lis}.
\end{problem}

In many applications we can only access the Gauss-Newton Hessian by computing its action on vectors, which involves one forward model evaluation and one adjoint model evaluation. In such a case, the \ppgnh can be approximated by finding the eigendecomposition \eqref{eq:eig} using either Krylov subspace algorithms \cite{Golub_2012} or randomized algorithms \cite{SVD_HMT_2011, SVD_Liberty_etal_2007}. %

The number of samples required to construct the global LIS depends on
the degree to which $H(x)$ (or its dominant eigenspace) varies over
the posterior. In Section~\ref{sec:algo}, we present an adaptive
construction procedure that automatically explores the directional
landscape of the likelihood.

\subsection{Posterior approximation}
\label{sec:var_redu}

By projecting the likelihood function onto the global likelihood-informed subspace (LIS), we obtain the approximate posterior
\begin{equation}
\tilde{\pi}(\param \vert \data) \propto \pi\left(\data  \vert  \projlis \param\right) \pi_0(\param),
\label{eq:llkd_proj}
\end{equation}
where $\Pi_r$ is a projector onto the global LIS.
This projector is self-adjoint with respect to the inner product induced by the prior precision matrix, and leads to a natural decomposition of the parameter space as $\paramspace = \paramspace_r \oplus \paramspace_\perp$, where $\paramspace_r = \range(\projlis)$ is the LIS and $\paramspace_\perp = \range(I - \projlis)$ is the complement subspace (CS).
This choice leads to a factorization of the prior distribution into the product of two distributions, one defined on the low-dimensional LIS and the other on the CS. This factorization is the key to our dimension reduction technique. 

\begin{definition}
\label{def:proj}
We define the projectors $\projlis$ and $I - \projlis$, and a corresponding parameter decomposition, as follows: 
\begin{enumerate}[(a)]
\item Suppose the LIS basis computed in Problem~\ref{pro:g_lis} is $\gbasis_r = \chol \gbasisv_r$, where $\gbasisv_r$ is orthonormal. Define the matrix $\Xi_r = \chol^{-\top} \gbasisv_r^{} $ such that $\Xi_r^\top \gbasis_r^{} = I_r^{}$. 
The projector $\projlis$ has the form  
\[
\projlis^{} =  \gbasis_r^{} \Xi_r^\top.
\] 
Choose $\gbasisv_\perp$ such that $[\gbasisv_r \ \gbasisv_\perp]$ forms a complete orthonormal system in $\real^{n}$. Then the projector $I - \projlis$ can be written as 
\[
I - \projlis^{} = \gbasis_\perp^{} \Xi_\perp^\top,
\]
where $\gbasis_\perp = \chol \gbasisv_\perp$ and $\Xi_\perp = \chol^{-\top} \gbasisv_\perp^{}$. 
\item Naturally, the parameter $\param$ can be decomposed as %
\[
\param = \projlis \param + (I - \projlis) \param ,
\]
where each projection can be represented as the linear combination of the corresponding basis vectors.
Consider the ``LIS parameter'' $\param_r$ and the ``CS parameter'' $\param_\perp$, which are the weights associated with the LIS basis $\gbasis_r$ and CS basis $\gbasis_\perp$, respectively. 
Then we can define the following pair of linear transformations between the parameter $\param$ and $(\param_r, \param_\perp)$:
\begin{equation}
\param = \left[\gbasis_r \; \gbasis_\perp\right] \left[\begin{array}{l} \param_r \\ \param_\perp \end{array}\right], \quad 
\left[\begin{array}{l} \param_r \\ \param_\perp \end{array}\right] = \left[\Xi_r \; \Xi_\perp \right]^\top \param.
\label{eq:tran}
\end{equation}
\end{enumerate}
\end{definition}

Figure \ref{fig:tran} illustrates the transformations between the parameter projected onto the LIS, $\projlis \param$, and the LIS parameter $\param_r$. The same relation holds for the transformations between $(I - \projlis) \param$ and the CS parameter $\param_\perp$. And as Definition \ref{def:proj} makes clear, $\projlis$ is an oblique projector.
\begin{figure}[h!]
\centerline{\includegraphics[width=0.5\textwidth]{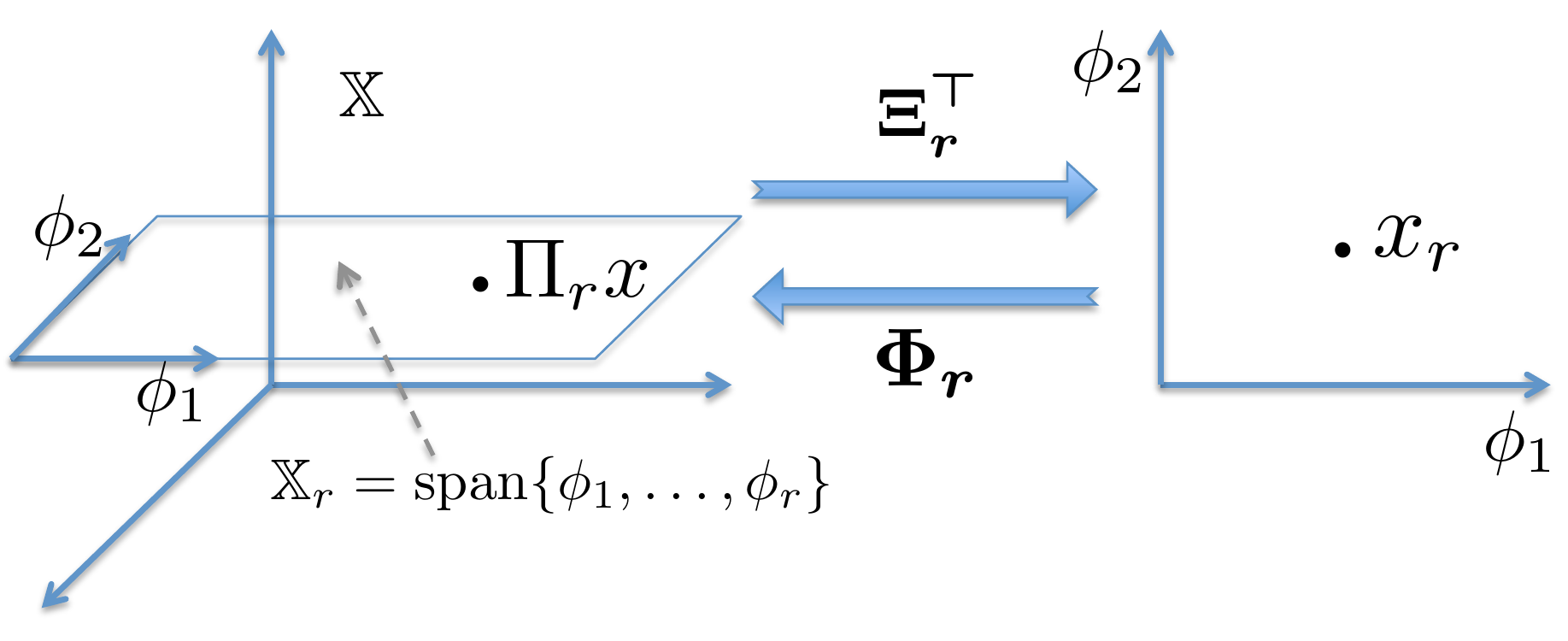}}
\caption{Illustration of the transformation between the parameter projected onto the LIS, $\projlis \param$, and the LIS parameter $\param_r$.}
\label{fig:tran}
\end{figure}

\begin{lemma}
\label{prior_fac}
Suppose we have $\param = \gbasis_r \param_r + \gbasis_\perp \param_\perp$ as defined in Definition \ref{def:proj}(b). Then the prior distribution can be decomposed as
\[
\pi_0(\param) = \pi_r(\param_r) \pi_\perp(\param_\perp),
\] 
where $\pi_r(\param_r) = \normal(\Xi_r^\top \prmean, I_r^{})$ and $\pi_\perp(\param_\perp) = \normal(\Xi_\perp^\top \prmean, I_\perp^{})$.
\end{lemma}

Following Definition \ref{def:proj}(b) and Lemma \ref{prior_fac}, the approximate posterior distribution \eqref{eq:llkd_proj} can be reformulated as
\begin{eqnarray*}
\tilde{\pi}(\param \vert \data) & \propto & \pi\left(\data  \vert  \projlis \param\right) \pi_r(\param_r) \pi_\perp(\param_\perp) \\
& = & \pi\left(\data  \vert  \gbasis_r \param_r \right) \pi_r(\param_r) \pi_\perp(\param_\perp).
\end{eqnarray*}
Applying the linear transformation from $\param$ to $(\param_r, \param_\perp)$ as defined in Equation \eqref{eq:tran}, we can rewrite the approximate posterior for the parameters $(\param_r, \param_\perp)$ as
\begin{equation}
\tilde{\pi}(\param_r, \param_\perp \vert \data) \propto \tilde{\pi}(\param \vert \data) \propto \tilde{\pi}(\param_r \vert \data) \pi_\perp(\param_\perp),
\label{eq:approx_pos}
\end{equation} 
which is the product of  the reduced posterior 
\begin{equation}
\tilde{\pi}(\param_r \vert \data) \propto \pi\left(\data \vert \gbasis_r \param_r\right) \pi_r(\param_r),
\label{eq:redu_pos}
\end{equation}
and the complement prior $\pi_\perp(\param_\perp)$. 
To compute a Monte Carlo estimate of the expectation of a function over the approximate posterior distribution
\eqref{eq:approx_pos}, we only need to sample the reduced posterior $\tilde{\pi}\left(\param_r  \vert  \data\right)$, since properties of the Gaussian complement prior $\pi_\perp(\param_\perp)$ are known analytically.

One can also combine MCMC samples from the reduced posterior $\tilde{\pi}\left(\param_r  \vert  \data\right)$ with independent samples from the complement prior $\pi_\perp (\param_\perp)$ to provide samples that are approximately drawn from the full posterior $\pi(\param \vert \data)$. By correcting these samples via importance weights or a Metropolis scheme, one would then obtain a sampling algorithm for the original full-space posterior. This idea is not pursued further here, and in the rest of this work we will emphasize the analytical properties of the complement prior $\pi_\perp(\param_\perp)$, using them to reduce the variance of Monte Carlo estimators.

\subsection{Reduced-variance estimators}

Suppose we have a function $h(\param)$ for which the conditional expectation over the approximate posterior \eqref{eq:llkd_proj}  
\begin{equation}
\expectapprox \left[h(\param)  \vert  \param_r\right] =
\int_{\paramspace_\perp} h(\gbasis_r \param_r + \gbasis_\perp \param_\perp ) \, \pi_0(\param_\perp) \, d \param_{\perp},
\label{eq:cond_exp}
\end{equation}
can be calculated either analytically or through some high-precision numerical quadrature scheme. Then variance reduction can be achieved as follows:
\begin{enumerate}
\item {\bf Subspace MCMC.} Use MCMC in the LIS to simulate a ``subspace Markov chain'' with target distribution $\tilde{\pi}(\param_r \vert \data)$ \eqref{eq:approx_pos}.  Any number of MCMC algorithms developed in the literature can be applied off-the-shelf, e.g., adaptive MCMC \cite{Haario_2001, AM_adapt_2006, Atchade_2006, Roberts_2007, Roberts_2009}, the stochastic Newton algorithm of \cite{Martin_2012}, or the Riemannian manifold algorithms of \cite{Girolami_2011}.
Since the dimension of the LIS can be quite small relative to the original parameter space, the subspace MCMC approach can yield lower sample correlations (better mixing) than applying any of these MCMC algorithms directly to the full posterior \eqref{eq:post}.  
\item {\bf Rao-Blackwellization.} We approximate $\expectpos[h] = \int_{\paramspace} h(\param) \pi(d\param \vert \data)$ by the expectation of the function $h(\param)$ over the approximate posterior $\tilde{\pi}(\param \vert \data)$, i.e., $\expectapprox[h] = \int_{\paramspace} h(\param) \tilde{\pi}(d\param \vert \data)$.
Given a set of subspace MCMC samples $\{\param_r^{(1)}, \ldots, \param_r^{(N)}\}$ where $\param_r^{(k)} \sim \tilde{\pi}_r(\param_r \vert \data)$, a Monte Carlo estimator of $\expectapprox[h]$ is given by
\begin{equation}
\MCRB = \frac{1}{N} \sum_{k=1}^{N} \expectapprox \left[h(\param)  \vert  \param_r^{(k)}\right].
\label{eq:RB_I}
\end{equation} 
As an application of the Rao-Blackwellization principle (see \cite{CR_1996} and references therein), the estimator \eqref{eq:RB_I} has a lower variance than the standard estimator
\begin{equation}
 \MC = \frac{1}{N} \sum_{k = 1}^{N}h(\param^{(k)}),
 \label{eq:standard_I}
\end{equation}
where $\param^{(k)} \sim \pi(\param \vert \data)$.
\end{enumerate}
This procedure mitigates many of the difficulties of posterior exploration in high dimensions, provided that the prior-to-posterior update is reasonably low rank. Variance reduction is achieved not only by increasing the effective sample size per MCMC iteration (via subspace MCMC), but also by reducing the variance of Monte Carlo estimators using Rao-Blackwellization. In effect, we argue that the Gaussian CS can be explored \textit{separately} and via a calculation \eqref{eq:cond_exp} that does not involve sampling. 

Note that while this procedure necessarily reduces the variance of a Monte Carlo estimator, it introduces bias since we replace the expectation over the full posterior $\expectpos[h]$ with an expectation over the approximate posterior $\expectapprox[h]$. Thus this variance reduction is particularly useful in situations where the variance of the estimator \eqref{eq:standard_I} derived from full-space MCMC samples is large compared with the bias, which is often the case for high-dimensional inverse problems.

Beyond variance reduction, subspace MCMC offers several additional computational advantages over MCMC methods applied to the full posterior directly: ({\romannumeral 1}) The storage requirement for saving subspace MCMC samples is much lower than that of an MCMC scheme that samples the full posterior. ({\romannumeral 2}) For MCMC methods where the proposal distribution involves operations with square root of the prior covariance matrix (e.g., the stochastic Newton \cite{Martin_2012} and preconditioned Crank-Nicolson \cite{BRSV_2008, CRSW_2012} techniques) the computational cost of handling the full prior covariance can be much higher than the computational cost of handling the reduced prior $\pi_r(\param_r)$, which has identity covariance. 

\medskip

The Monte Carlo estimator \eqref{eq:RB_I} can be further simplified if the function of interest $h(\param)$ can be expressed as either the product or the sum of two separate functions, $h_r(\param_r)$ and $h_\perp(\param_\perp)$, defined on the LIS and CS, respectively. 
In the multiplicative case $h(\param) = h_r(\param_r)h_\perp(\param_\perp)$, the conditional expectation \eqref{eq:cond_exp} can be written as
\[
\expectapprox \left[h(\param)  \vert  \param_r\right] = h_r(\param_r) \int_{\paramspace_\perp}  h_\perp(\param_\perp) \pi_0(d\param_\perp) .
\]
In the additive case $h(\param) = h_r(\param_r) + h_\perp(\param_\perp)$, it can be written as
\[
\expectapprox \left[h(\param)  \vert  \param_r\right] = h_r(\param_r) + \int_{\paramspace_\perp}  h_\perp(\param_\perp) \pi_0(d\param_\perp).
\]
Thus the expectation $\expectapprox[h]$ can be decomposed either into the product (in the multiplicative case) or the sum (in the additive case) of the pair of expectations
\begin{eqnarray}
\expectapprox[h_r] & = & \int_{\paramspace_r} h_r(\param_r) \pi(d\param_r \vert \data), \label{eq:exp_r} \\
\expectapprox[h_\perp] & = & \int_{\paramspace_\perp} h_\perp(\param_\perp) \pi_0(d\param_\perp), \label{eq:exp_perp}
\end{eqnarray}
which are associated with the LIS and CS, respectively. 
The expectation in \eqref{eq:exp_r} can be computed by the subspace MCMC methods described above, whereas the expectation in \eqref{eq:exp_perp} is computed analytically or through high-order numerical integration. 

Now we give two particularly useful examples of the analytical treatment of the complement space.
\begin{example}[Reduced variance estimator of the posterior mean] 
\label{ex:mean}
Suppose we have obtained the empirical posterior mean $\tilde{\mu}_r$ of the reduced parameter $\param_r$ using subspace MCMC. 
The resulting reduced-variance estimator of the posterior mean is
\begin{equation*}
\expectapprox[\param] = \gbasis_r \tilde{\mu}_r + \projcs \prmean = \gbasis_r \tilde{\mu}_r + (I - \projlis) \prmean.
\end{equation*}
\end{example}

\begin{example}[Reduced variance estimator of the posterior covariance] 
\label{ex:cov}
Suppose we have the empirical posterior covariance $\tilde{\Gamma}_r$ of the reduced parameter $\param_r$, estimated using subspace MCMC. 
The resulting reduced-variance estimator of the posterior covariance is
\begin{eqnarray*}
\mathbb{C}\mathrm{ov}_{\tilde{\pi}}[\param] & = & \gbasis_r^{} \tilde{\Gamma}_r^{} \gbasis_r^\top + \projcs^{} \prcov^{} \projcs^\top \nonumber \\
& = & \prcov + \gbasis_r \left( \tilde{\Gamma}_r^{} - I_r \right) \gbasis_r^\top .
\end{eqnarray*}
\end{example}

\subsection{Algorithms for the LIS}
\label{sec:algo}
Constructing the global LIS requires a set of posterior samples. Since the computational cost of solving Problem \ref{pro:l_lis} for any sample is much greater than the cost of evaluating the forward model, we wish to limit the number of samples used in Problem \ref{pro:g_lis} while ensuring that we adequately capture the posterior variation of the ppGNH. Thus we choose samples using the following adaptive procedure. 

\begin{algo}[Global LIS construction using subspace MCMC] 
\label{algo:subspace}
First, compute the posterior mode $\param_{map} \in \paramspace$. Set the initial sample set for Problem~\ref{pro:g_lis} to $\mathcal{X}^{(1)}= \{ \param_{map} \}$. Solve Problem~\ref{pro:g_lis} to find $\gbasisv_r^{(1)}$, the initial LIS basis $\gbasis_{r}^{(1)}$, and its left-inverse $\Xi_{r}^{(1)}$.\footnote{The dimension of the global LIS can vary at each iteration. Let $r(k)$ denote the dimension of the global LIS at iteration $k$. To be precise, we should then write $\gbasis_{r(k)}^{(k)}$ and $\Xi_{r(k)}^{(k)}$, but for brevity we will simplify notation to $\gbasis_{r}^{(k)}$ and $\Xi_{r}^{(k)}$ when possible.} 
Initialize a subspace Markov chain with initial state $\Xi_{r}^{(1)\top} \param_{map}$, which is the posterior mode projected onto the LIS. At any subsequent step $k \geq 1$, the following procedure is used to adaptively enrich the LIS:
\begin{enumerate}
\item \label{adapt:step1} {\bf Subchain simulation.}  Simulate the $r(k)$-dimensional subspace MCMC chain for $L$ iterations, so that the last state of this chain, denoted by $\theta$, is uncorrelated with its initial state. Then $\theta$ transformed back to the original parameter space, $(\gbasis_r^{(k)} \theta)$, is used as the next sample point. Enrich the sample set to $\mathcal{X}^{(k+1)} = \mathcal{X}^{(k)} \cup \{ \gbasis_r^{(k)} \theta \}$.

\item {\bf LIS construction.} Solve Problem \ref{pro:g_lis} with the sample set $\mathcal{X}^{(k+1)}$. Then update the LIS basis to $\gbasis_r^{(k+1)}$ and $\, \Xi_r^{(k+1)}$. Set the initial state of the next subspace MCMC chain to $\Xi_r^{(k+1)\top} \gbasis_r^{(k)} \theta$.

\item \label{step:conv} {\bf Convergence checking.} Terminate the adaptation if a pre-specified maximum allowable number of Hessian evaluations is exceeded, or if the weighted subspace distance in Definition \ref{def:conv} falls below a certain threshold. 
Otherwise, set $k \leftarrow k+1$ and return to Step (\ref{adapt:step1}).
\end{enumerate}
\end{algo}

The convergence criterion in step (\ref{step:conv}) is based on an incremental distance between likelihood-informed subspaces. The distance penalizes changes in the dominant directions (those with large eigenvalues $\gamma$) more heavily than changes in the less important directions (those with smaller $\gamma$). 

\begin{definition}[Weighted subspace distance]
\label{def:conv}
At iteration $k$, define the basis/weights pair $\mathcal{Y}^{(k)} = \{\gbasisv_r^{(k)}, D^{(k)}\} $, where $\gbasisv_r^{(k)}$ is the orthonormal LIS basis from Problem \ref{pro:g_lis} and $D^{(k)}_{ij} = \delta_{ij} (\hat{\gamma}_i^{(k)})^{\frac{1}{4}}$ is a diagonal matrix consisting of normalized weights  
\[
\hat{\gamma}_i^{(k)} = \frac{\gamma_i^{(k)}}{\sum_{j = 1}^{r(k)}\gamma_j^{(k)}}, \quad j = 1 \ldots r(k),
\]
computed from the eigenvalues $\{\gamma_1^{(k)}, \ldots, \gamma_{r(k)}^{(k)}\}$ of  Problem \ref{pro:g_lis}.
For two adjacent steps $k$ and $k+1$, we compute the weighted subspace distance of \cite{Distance_Li_2009}, which has the form
\begin{equation}
\mathcal{D}\left( \mathcal{Y}^{(k)} ,\mathcal{Y}^{(k+1)} \right) = \sqrt{
1 - \left\| \left( \gbasisv_{r(k)}^{(k)} D^{(k)} \right)^\top \left( \gbasisv_{r(k+1)}^{(k+1)} D^{(k+1)} \right) \right\|_{F}^2} \, .
\label{eq:lis_conv}
\end{equation}
\end{definition}

\medskip

Note that in Step (i) of Algorithm 8, we construct the global LIS by always sampling in an adaptively enriched subspace. This offers computational benefits, since the MCMC exploration is always confined to a lower dimensional space. However, a potential problem with this approach is that it might ignore some directions that are also data-informed. A more conservative approach would be to introduce a {\it conditional update} at the end of each subchain simulation: perform Metropolized independence sampling in the current CS using the complement prior as proposal. This would enable the subchain to explore the full posterior, but would result in higher-dimensional sampling when constructing the LIS. In our numerical examples, described below, no conditional updates were required for good performance; constructing the LIS using samples from the full posterior and using the subspace approach gave essentially the same results.
Of course, one could also simply employ a standard MCMC algorithm to sample the full posterior, and then construct the LIS using the resulting posterior samples. However, the efficiency of the MCMC algorithm in this case will be affected by the dimensionality of the problem.

\section{Example 1: Elliptic PDE}
\label{sec:elliptic}

Our first example is an elliptic PDE inverse problem used to demonstrate ({\romannumeral 1}) construction of the LIS and the impact of mesh refinement; ({\romannumeral 2}) the application of low-rank posterior mean and variance estimators; and ({\romannumeral 3}) changes in the LIS with varying amounts of observational data.

\subsection{Problem setup}
\label{sec:ellpitic_setup}
Consider the problem domain $\Omega = [0, 3]\times [0, 1]$, with boundary $\partial \Omega$. We denote the spatial coordinate by $s \in \Omega$.
Consider the permeability field $\kappa(s)$, the pressure field $p(s)$, and sink/source terms $f(s)$. 
The pressure field for a given permeability and source/sink configuration is governed by the Poisson equation 
\begin{equation}
\left\{ 
\begin{array}{lcll}
-\nabla \cdot \left( \kappa(s) \nabla p(s) \right) & = & f(s), & s \in \Omega \\
\langle \kappa(s) \nabla p(s), \vec{n}(s) \rangle & = & 0, & s \in \partial \Omega
\end{array}
\right.
\label{eq:forward_e}
\end{equation}
where $\vec{n}(s)$ is the outward normal vector on the boundary. To make a well-posed boundary value problem, a further boundary condition
\begin{equation}
\int_{\partial \Omega} p(s) ds = 0,
\label{eq:bnd}
\end{equation}
is imposed.
The source/sink term $f(s)$ is defined by the superposition of four weighted Gaussian plumes with standard deviation (i.e., spatial width) $0.05$, centered at four corners $[0, 0]$, $[3, 0]$, $[3, 1]$ and $[0, 1]$, with weights $\{1, 2, 3, -6\}$.
The system of equations (\ref{eq:forward_e}) is solved by the finite element method with $120\times 40$ bilinear elements. 

The discretized permeability field $\kappa$ is endowed with a log-normal prior distribution, i.e.,
\begin{equation}
\label{eq:prior_e}
\kappa = \exp(x), \; {\rm and} \; x \sim \normal\left(0, \prcov \right),
\end{equation}
where the covariance matrix $\prcov$ is defined through an anisotropic exponential covariance kernel
\begin{equation}
\label{eq:corr}
\mathbb{C}{\rm ov} \left (x(s), x(s^\prime) \right ) = 
\sigma_u^2 \exp \left( - \frac{\left( (s - s^{\prime})^\top \Sigma^{-1} (s-s^{\prime})\right)^{\frac12}}
{ s_0 } \right), 
\end{equation}
for $s, s^{\prime} \in \Omega$.
In this example, we set the anisotropic correlation tensor to
\[
\Sigma = \left[\begin{array}{rr} 0.55 & -0.45\\ -0.45 & 0.55 \end{array}\right],
\]
the prior standard deviation to $\sigma_u = 1.15$, and the correlation length to $s_0 = 0.18$.
The ``true'' permeability field is a realization from the prior distribution. The true permeability field, the sources/sinks, the simulated pressure field, and the synthetic data are shown in Figure \ref{fig:setup_e}.
\begin{figure}[h!]
\centerline{\includegraphics[width=\textwidth]{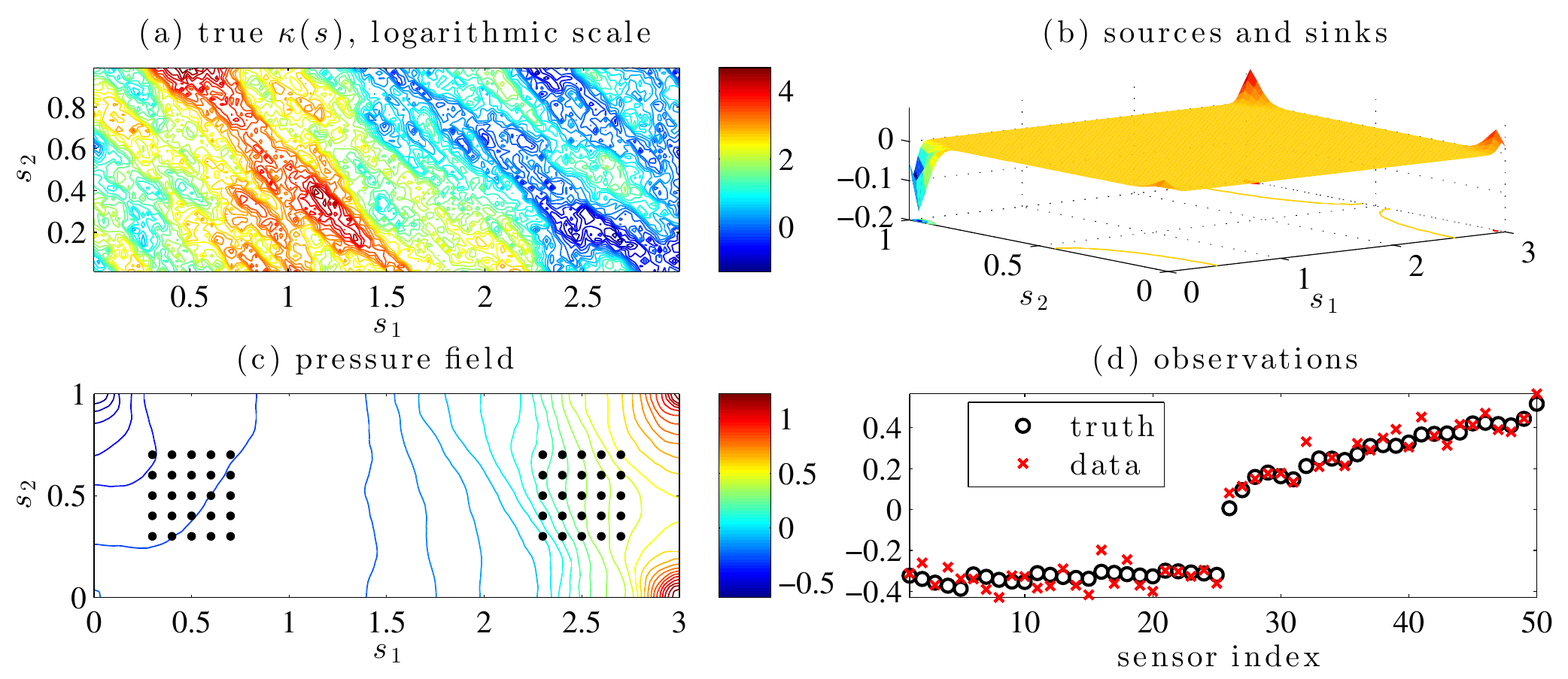}}
\caption{Setup of the elliptic inversion example. (a) ``True'' permeability field. (b) Sources and sinks. (c) Pressure field resulting from the true permeability field, with measurement sensors indicated by black circles. (d) Data $y$; circles represent the noise-free pressure at each sensor, while crosses represent the pressure observations corrupted with measurement noise.}
\label{fig:setup_e}
\end{figure}

\begin{figure}[h!]
\centerline{\includegraphics[width=0.9\textwidth]{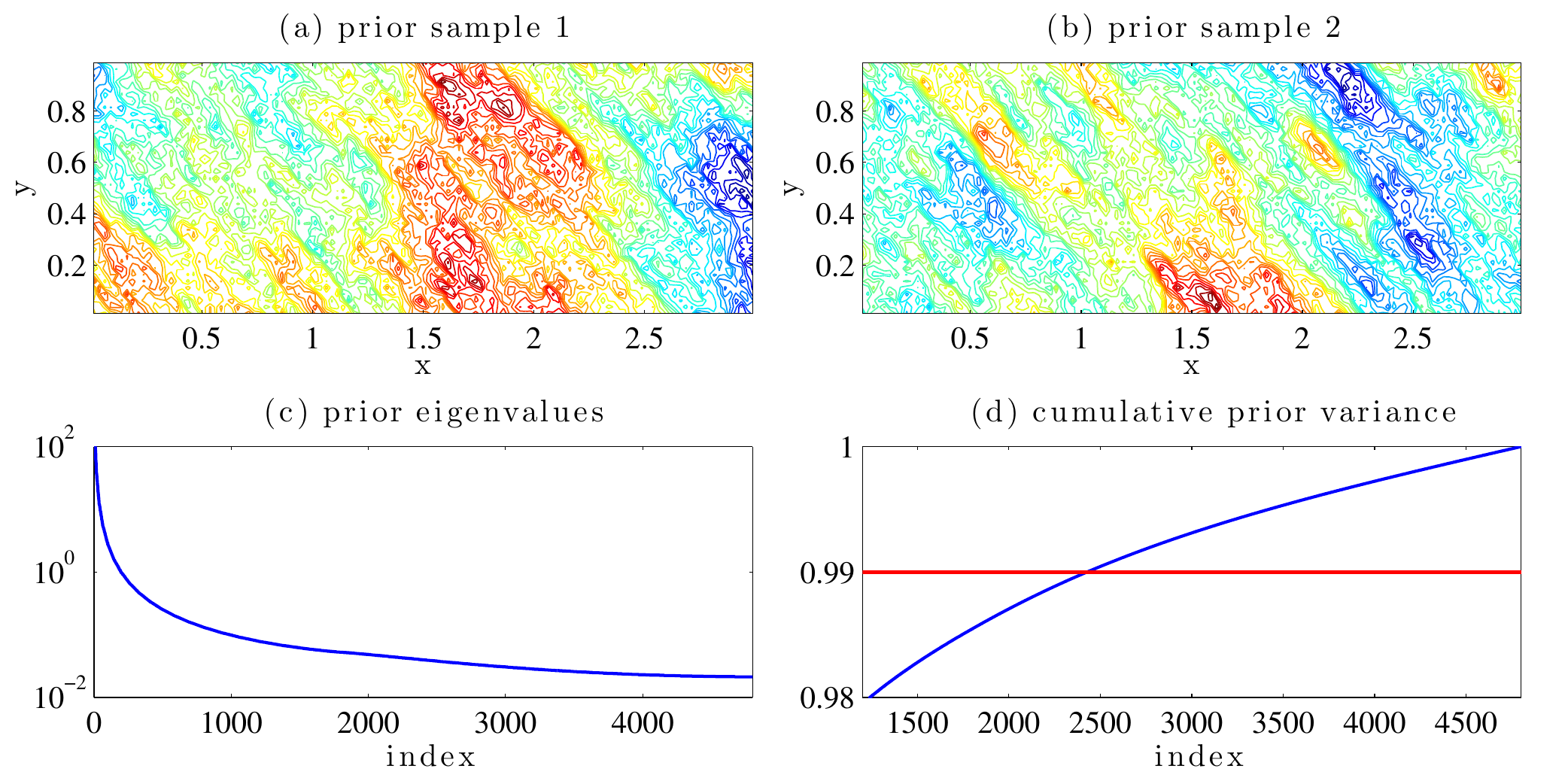}}
\caption{Prior samples and eigenspectrum of the prior covariance. (a) and (b): Two samples drawn from the prior. (c) Prior covariance spectrum, eigenvalues versus index number. (d) Cumulative energy (integrated prior variance) over a subset of the eigenspectrum, shown in blue; the red line represents the $99\%$ energy truncation threshold.}
\label{fig:prior_e}
\end{figure}

Partial observations of the pressure field are collected at $50$ measurement sensors as shown by the black dots in Figure \ref{fig:setup_e}(c).  The observation operator $M$ is simply the corresponding ``mask'' operation.  This yields observed data  $y \in \real^{50}$ as
\[
y = M p(s) + e , 
\]
with additive error $e \sim \normal(0, \sigma^2 I_{50})$.
The standard deviation $\sigma$  of the measurement noise is prescribed so that the observations have signal-to-noise ratio 10, where the signal-to-noise ratio is defined as $\max_s\{p(s)\}/\sigma$.  
The noisy data are shown in Figure \ref{fig:setup_e}(d).

Figure \ref{fig:prior_e} shows two draws from the prior distribution,  the eigenspectrum of the prior covariance, and the cumulative prior variance integrated over $\Omega$ (i.e., the running sum of the prior covariance eigenvalues).
In order to keep $99\%$ percent of the energy in the prior, $2427$ eigenmodes are required.
Because of this slow decay of the prior covariance spectrum, {\it a priori} dimension reduction based on a truncated eigendecomposition of the prior covariance (as described in \cite{Marzouk_2009}) would be very inefficient for this problem. Information carried in high-frequency eigenfunctions cannot be captured unless an enormous number of prior modes are retained; thus, a better basis is required.

\subsection{LIS construction}
\label{sec:elliptic_lis}

Now we demonstrate the process of LIS construction using Algorithm \ref{algo:subspace}, and the structure of the LIS under mesh refinement. 
To compute the LIS, we run Algorithm \ref{algo:subspace} for $500$ iterations, using adaptive MALA \cite{Atchade_2006} to simulate each subchain with length $L = 200$.
We choose the truncation thresholds $\tau_{loc} = \tau_{g} = 0.1$. 
To explore the dimensionality and structure of the LIS versus mesh refinement, we carry out the same numerical experiment on a $60\times 20$ coarse grid, a $120 \times 40$ intermediate grid, and a $180\times 60$ fine grid.
The dimension of the LIS versus number of iterations, the evolution of the convergence diagnostic \eqref{eq:lis_conv}, and the generalized eigenvalues after 500 iterations---for each level of grid refinement---are shown in Figure \ref{fig:SNR10_conv}. Also, Figure \ref{fig:SNR10_basis} shows the first five LIS basis vectors for each level of discretization.

\begin{figure}[h!]
\centerline{\includegraphics[width=0.8\textwidth]{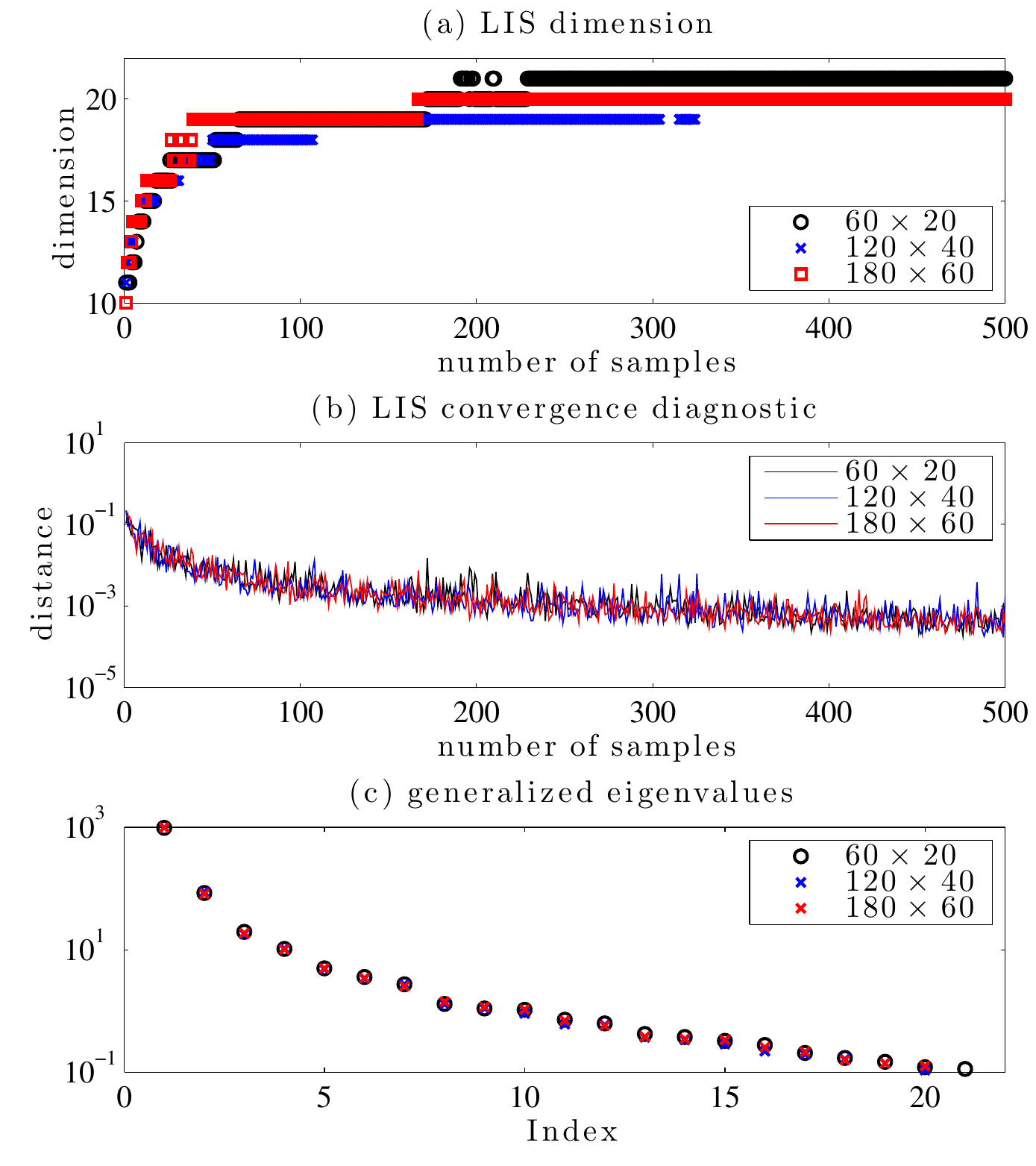}}
\caption{The dimension of the LIS and the convergence diagnostic \eqref{eq:lis_conv} versus the number of samples used in the adaptive construction. Black, blue, and red markers represent the $60\times 20$ grid, the $120\times 40$ grid, and the $180\times 60$ grid, respectively. Subplot (a) shows the dimension of the LIS; subplot (b) shows the weighted distance between successive subspaces; and subplot (c) shows the generalized eigenvalues $\gamma_i^{(k)}$ after $k=500$ iterations.}
\label{fig:SNR10_conv}
\end{figure}

\begin{figure}[h!]
\centerline{\includegraphics[width=\textwidth]{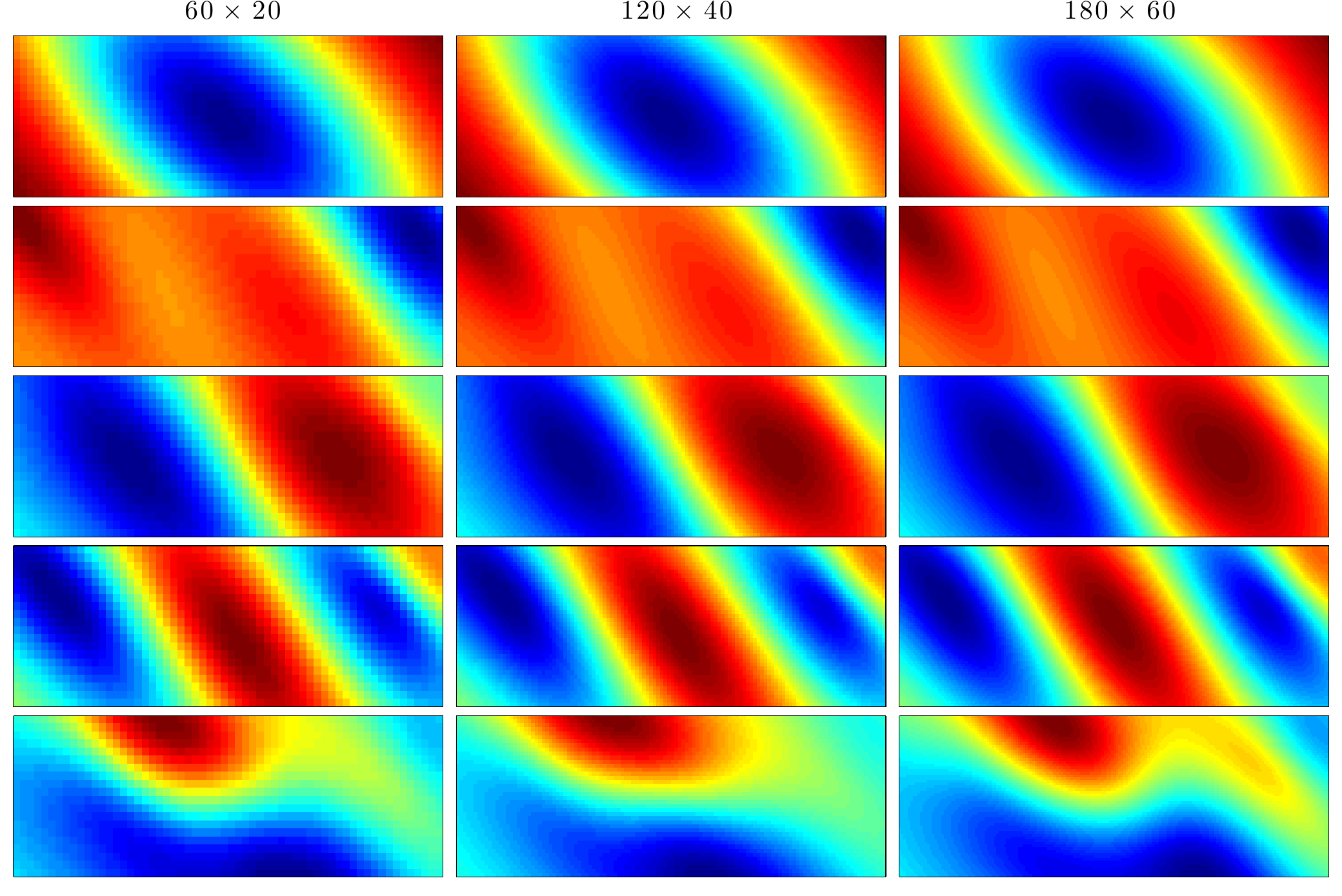}}
\caption{The first five LIS basis vectors (columns of $\gbasis_5$) for different levels of discretization of the inversion parameters $x$. In the figure, columns 1--3 correspond to the $60\times 20$ grid, the $120\times 40$ grid, and the $180\times 60$ grid, respectively. The basis vectors in each column are ordered top to bottom by decreasing eigenvalue.}
\label{fig:SNR10_basis}
\end{figure}

As shown in Figure \ref{fig:SNR10_conv}(a), the dimension of the LIS changes rapidly in the first 100 iterations, then it stabilizes. Change in the dimension reflects the fact that the log-likelihood Hessian $H(x)$ varies locally in this non-Gaussian problem. We also observe that the $60\times 20$ grid has a slightly larger final LIS dimension than the two refined grids: 
at the end of the adaptive construction, the LIS of the $60\times 20$ grid has dimension $21$, while the $120\times 40$ grid and the $180\times 60$ grid yield LIS dimensions of $20$.
This effect may be ascribed to larger discretization errors in the $60\times 20$ grid.

The weighted distance \eqref{eq:lis_conv} between each adjacent pair of likelihood-informed subspaces is used as the convergence diagnostic during the construction process.
With any of the three discretizations, the weighted subspace distance at the end of adaptive construction is several orders of magnitude lower than at the beginning, as shown in Figure \ref{fig:SNR10_conv}(b).
We also observe that the rates of convergence of this diagnostic are comparable for all three levels of discretization. These figures suggest that while local variation of the Hessian is important in this problem (e.g., the dimension of the LIS doubles over the course of the iterations), much of this variation is well-explored after 100 or 200 iterations of Algorithm \ref{algo:subspace}. 

Since the forward model converges with grid refinement, we expect that the associated LIS should also converge.  
The generalized eigenvalues for all three grids are shown in Figure \ref{fig:SNR10_conv}(c), where the spectra associated with all three subspaces have very similar values.
And as shown in Figure \ref{fig:SNR10_basis}, the leading LIS basis vectors $\{ \varphi_1, \ldots, \varphi_5\}$ have similar shapes for all three levels of grid refinement. Refinement leads to slightly more structure in $\varphi_5$, but the overall mode shapes are very close.

\subsection{Estimation of the posterior mean and variance}
\label{sec:elliptic_meanvar}

With an LIS in hand, we apply the variance reduction procedure described in Section \ref{sec:var_redu} to estimate the posterior mean and variance of the permeability field. Calculations in this subsection use the $120\times 40$ discretization of the PDE and inversion parameters.  

We first demonstrate the sampling performance of subspace MCMC, where we use adaptive MALA \cite{Atchade_2006} to sample the LIS-defined reduced posterior $\tilde{\pi}( x_r \vert y)$ \eqref{eq:redu_pos}.
We compare the results of subspace MCMC with the results of Hessian-preconditioned Langevin MCMC 
applied to the full posterior $\pi(x \vert y)$ \eqref{eq:post} (referred to as ``full-space MCMC'' hereafter).
The latter MCMC scheme results from an explicit discretization of the Langevin SDE, preconditioned by the inverse of the log-posterior Hessian evaluated at the posterior mode (see \cite{CLM_2014} for details). 
Note that we cannot precondition the full-dimensional Langevin SDE by the empirical posterior covariance as in adaptive MALA because of the high parameter dimension ($n=4800$). In this setup, subspace MCMC and full-space MCMC require the same number of forward model and gradient evaluations for a given number of MCMC iterations.

To examine sampling performance, the autocorrelation of the log-likelihood function and the autocorrelations of the parameters projected onto the first, third, and fifth LIS basis vectors are used as benchmarks. These results are shown in Figure \ref{fig:elliptic_auto}. We run both algorithms for $10^{6}$ iterations and discard the first half of the chains as burn-in.
The top row of Figure \ref{fig:elliptic_auto} shows these benchmarks for both samplers. For all four benchmarks, subspace MCMC produces a faster decay of autocorrelation as a function of sample lag---i.e., a lower correlation between samples after any given number of MCMC steps. %

Furthermore, as discussed in Section \ref{sec:var_redu}, even though the same number of forward model evaluations are required by subspace MCMC and full-space MCMC for a given number of samples, the computational cost of operations involving the square root of the prior covariance---used in sampling and evaluating the proposal distribution---can be much higher for full-space MCMC than subspace MCMC. In this test case, running subspace MCMC for $10^6$ iterations cost $2.1\times 10^4$ seconds of CPU time, while running full-space MCMC for the same number of iterations took $2.6\times10^5$ seconds. 
To incorporate this cost difference, the second row of Figure \ref{fig:elliptic_auto} shows the autocorrelation of the four benchmark quantities as a function of CPU time rather than sample lag. Here, we immediately observe that the autocorrelation per CPU time is further reduced by using subspace MCMC. 

Of course, recall that to construct the LIS we simulated Algorithm \ref{algo:subspace} for $500$ iterations. This costs roughly $2.2\times 10^4$ seconds of CPU time, which is only $8.5\%$ of the time required to run full-space MCMC for $10^6$ steps.
Therefore subspace MCMC, including the cost of LIS construction, takes less time to produce a given number of samples than full-space MCMC \textit{and} these samples are less correlated---i.e., of higher quality. 

\begin{figure}[h!]
\centerline{\includegraphics[width=\textwidth]{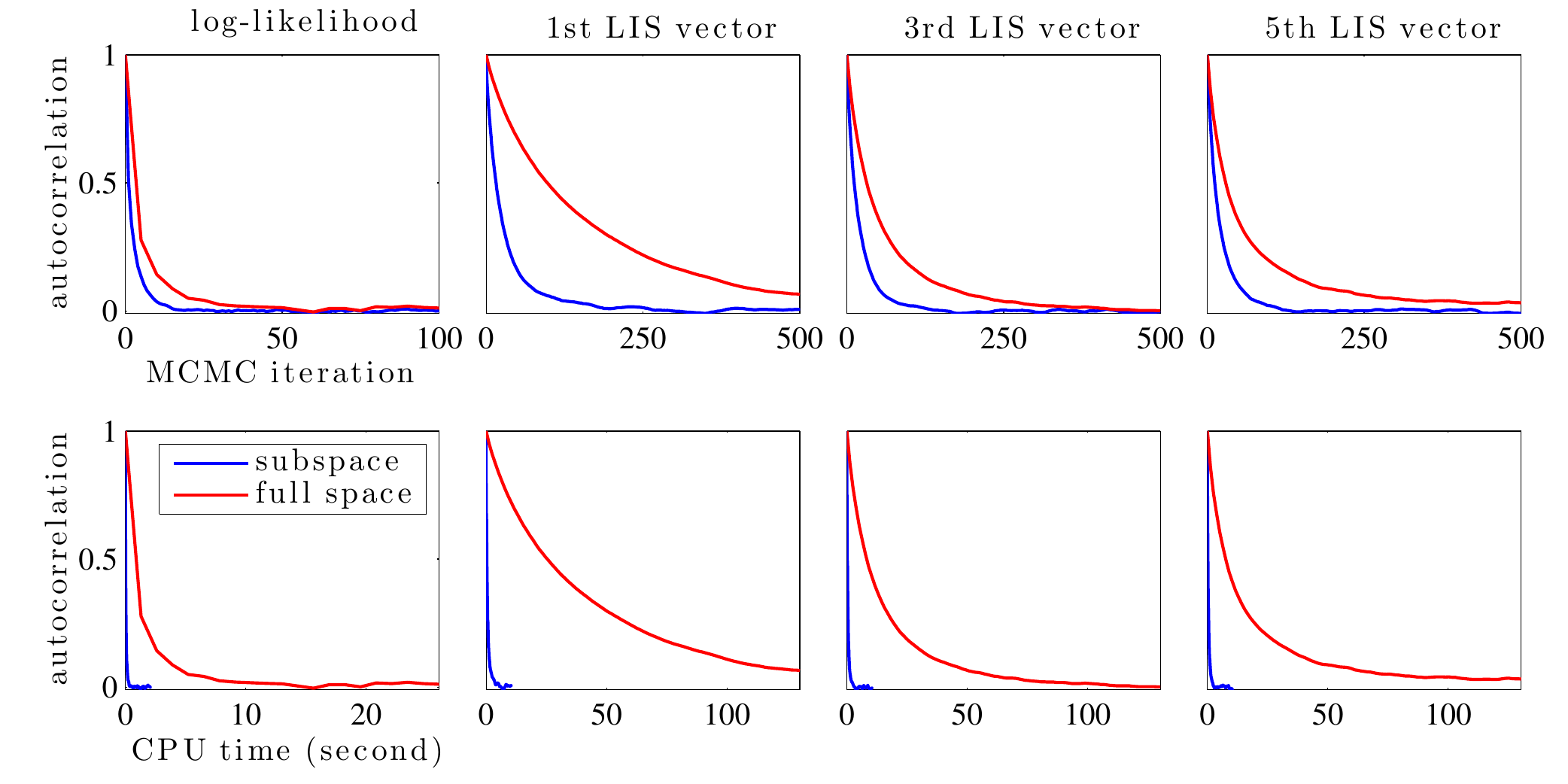}}
\caption{Autocorrelations of various benchmarks: blue line is subspace MCMC and red line is full-space MCMC. Column 1: log-likelihood function. Column 2: parameters projected onto the first LIS basis vector.  Column 3: parameters projected onto the third LIS basis vector.  Column 4: parameters projected onto the fifth LIS basis vector.  Top row: Autocorrelation as a function of sample lag. Bottom row: Autocorrelation as a function of sample lag, where the latter is measured via CPU time.}
\label{fig:elliptic_auto}
\end{figure}

\begin{figure}[h!]
\centerline{\includegraphics[width=\textwidth]{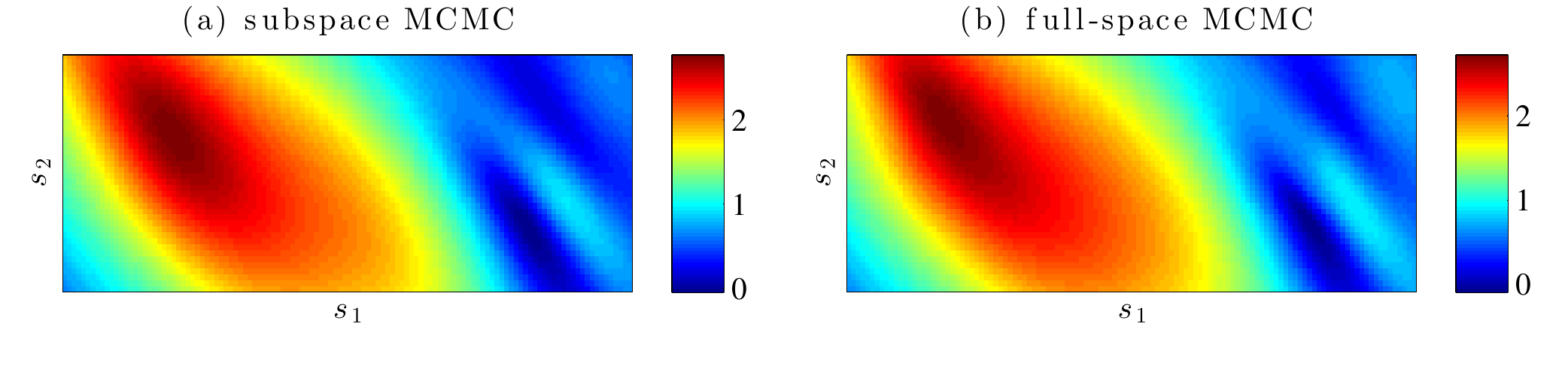}}
\caption{Estimates of posterior mean: (a) using subspace MCMC, (b) using full-space MCMC. }
\label{fig:elliptic_post_mean}
\end{figure}

\begin{figure}[h!]
\centerline{\includegraphics[width=\textwidth]{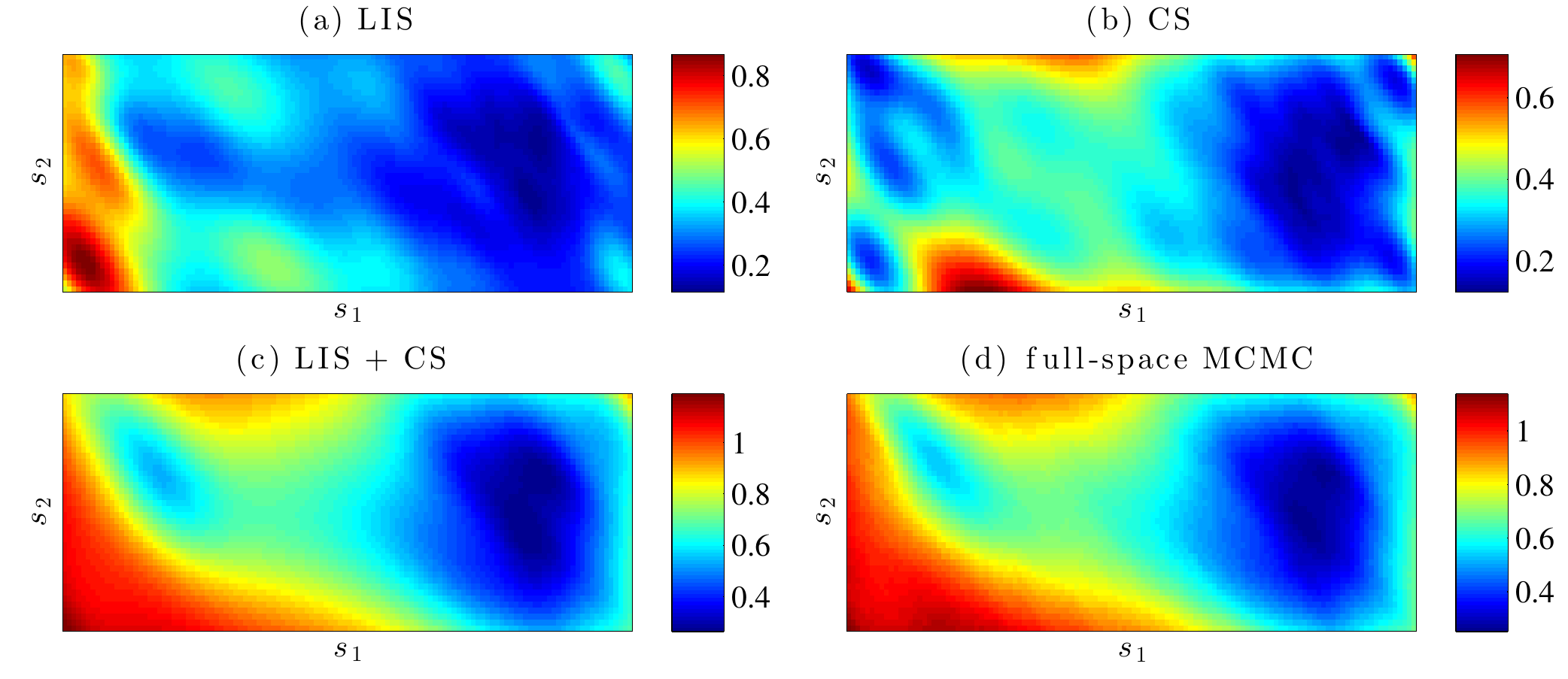}}
\caption{Estimation of the posterior variance: (a) empirical estimate using MCMC in the LIS; (b) analytical evaluation in the CS; (c) combined LIS + CS estimate; (d) for comparison, estimation using full-space MCMC. }
\label{fig:elliptic_post_var}
\end{figure}

We now compare reduced-variance estimates of the posterior mean and variance (obtained with subspace MCMC) with estimates obtained via full-space MCMC. 
The results are shown in Figures \ref{fig:elliptic_post_mean} and \ref{fig:elliptic_post_var}.
Full-space MCMC and subspace MCMC yield very similar mean and variance estimates.
Figures \ref{fig:elliptic_post_var}(a) and (b) distinguish the two components of the Rao-Blackwellized variance estimates described in Example~\ref{ex:cov}. Variance in the LIS, shown in Figure~\ref{fig:elliptic_post_var}(a), is estimated from MCMC samples, while variance in the CS, shown in Figure~\ref{fig:elliptic_post_var}(b), is calculated analytically from the prior and the LIS projector. The sum of these two variance fields is shown in Figure~\ref{fig:elliptic_post_var}(c), and it is nearly the same as the full-space result in Figure~\ref{fig:elliptic_post_var}(d). 
In the central part of the domain where measurement sensors are not installed, we can observe that the variance is larger in the CS than in the LIS, and hence this part of the domain is prior-dominated. 
In the right part of the domain, the variance is less prior-dominated, since this region is covered by observations.

\subsection{The influence of data}

The amount of information carried in the data affects the dimension and structure of the LIS. 
To demonstrate the impact of the data, we design a case study where different likelihood-informed subspaces are constructed under various observational scenarios. The same stationary groundwater problem defined in Section \ref{sec:ellpitic_setup} is employed here. 
For the sake of computational efficiency, the problem domain $\Omega = [0, 3]\times [0, 1]$ is discretized by a slightly coarser $72\times 24$ mesh. And to provide a stronger impulse to the groundwater system, the source/sink terms used in this example are different from those used in Sections \ref{sec:ellpitic_setup}--\ref{sec:elliptic_meanvar}.
Along the boundary of the domain $\Omega$, we evenly distribute a set of sources with a distance of $0.5$ between the source centers.
Two sinks are placed in the interior of the domain at locations $[0.5, 1]$ and $[2, 0.5]$.
Each source has weight $1$, while each sink has weight $3.5$.
We distributed sensors evenly over the domain $[0, 1]\times [0, 1] \cup [2, 3]\times [0, 1]$; starting with an inter-sensor spacing of $1/3$, we incrementally refine the sensor distribution with spacings of  $1/6$, $1/12$, and $1/24$. 
This results in four different data sets, containing the noisy readings of $32$, $98$, $338$, and $1250$ sensors, respectively. 
The true permeability field, the sources/sinks, the simulated pressure field, and sensor distributions are shown in Figure \ref{fig:setup_data_ref}.
\begin{figure}[h!]
\centerline{\includegraphics[width=\textwidth]{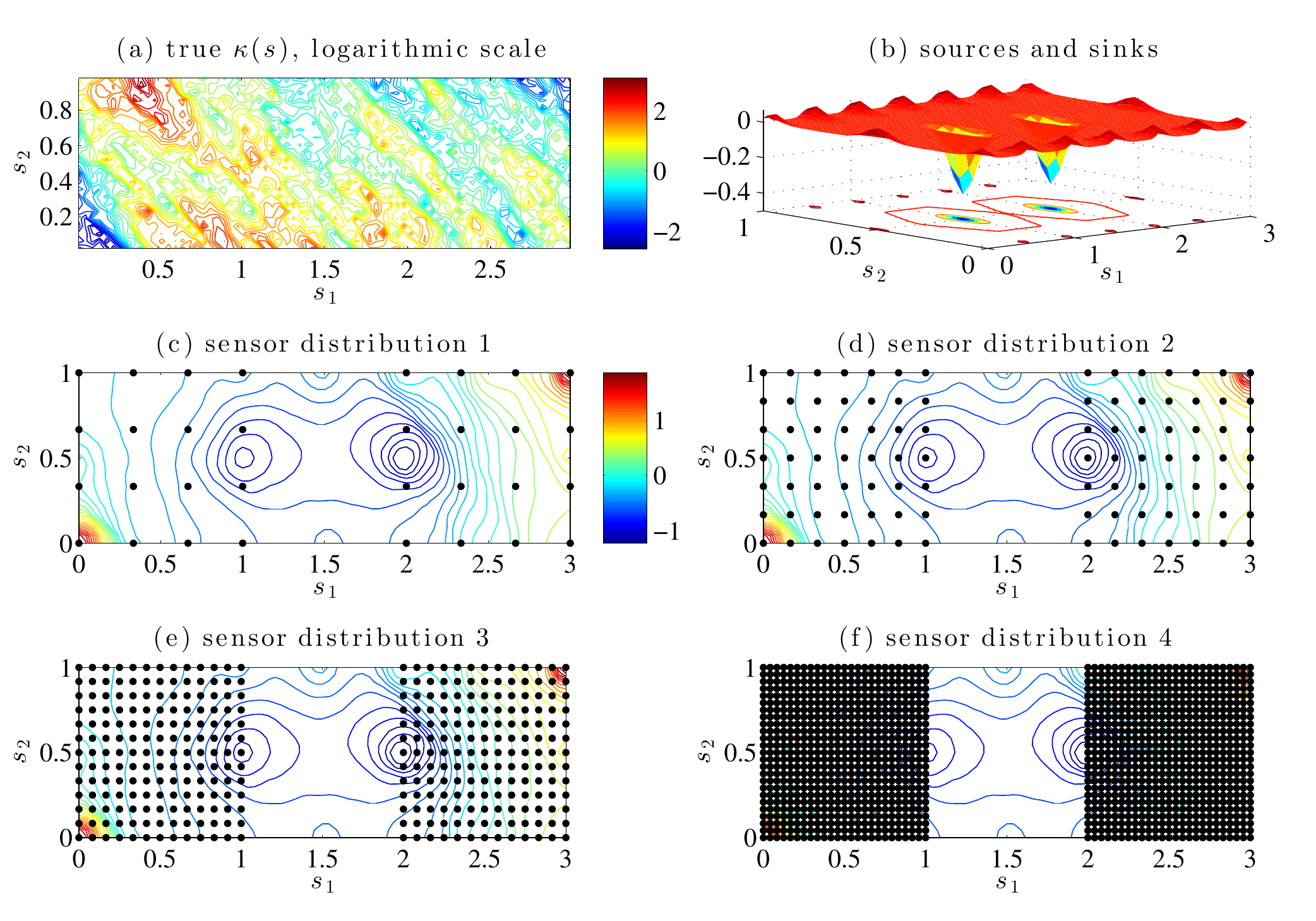}}
\caption{Setup of the elliptic inversion example for testing the influence of data. (a) True permeability field. (b) Sources and sinks. (c)--(f) Pressure field resulting from the permeability field defined in (a), and sensor distributions (black dots) for data sets 1--4.}
\label{fig:setup_data_ref}
\end{figure}

As in Section~\ref{sec:elliptic_lis}, we run Algorithm \ref{algo:subspace} for $500$ iterations to construct the LIS, using subchains of length $L = 200$. For data sets 1--4, the resulting LISs have dimension $24$, $34$, $50$, and $83$, respectively.
The generalized eigenvalue spectrum for each LIS is shown in Figure \ref{fig:data_ref_eig}. We note that the eigenvalues decay more slowly with increasing amounts of data. This behavior is expected; since the generalized eigenvalues reflect the impact of the likelihood, relative to the prior, more data should lead to more directions where the likelihood dominates the prior.

Since the sensors for all four data sets occupy the same area of the spatial domain, we expect that the four likelihood-informed subspaces should share a similar low frequency structure. However, the high frequency structures in each LIS might differ from each other under refinement of the sensor distribution. Thus the LIS basis vectors corresponding to the largest eigenvalues should share a similar pattern, while the LIS basis vectors corresponding to the relatively small eigenvalues might have different patterns. We observe this effect in the numerical experiments carried out here; Figure \ref{fig:data_ref_basis} shows the first and fifteenth LIS basis vector for each of the data sets.

\begin{figure}[h!]
\centerline{\includegraphics[width=0.5\textwidth]{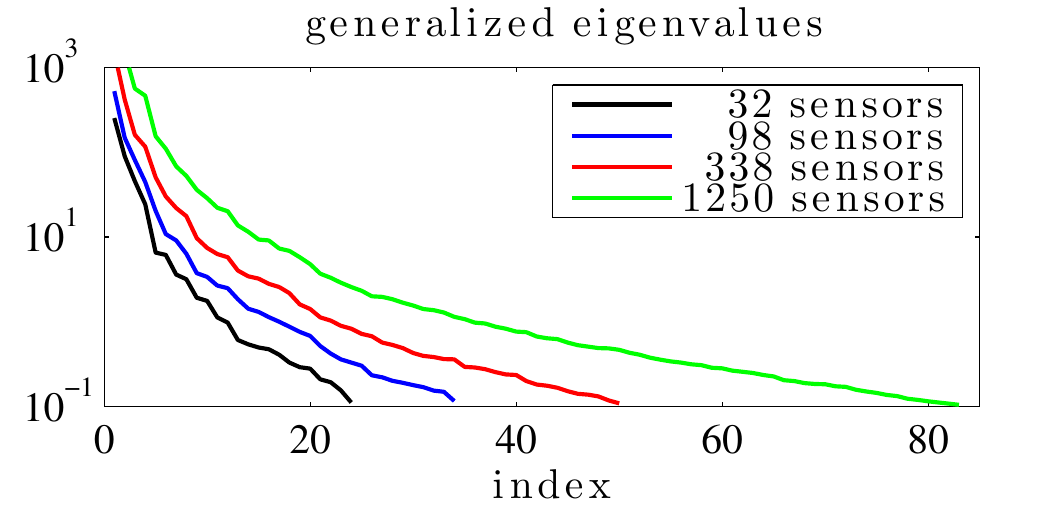}}
\caption{Generalized eigenvalues associated with the likelihood-informed subspace under refinement of the observations. The black, blue, red, and green lines show eigenvalues for data sets 1--4, with 32, 98, 338, and 1250 sensors, respectively.}
\label{fig:data_ref_eig}
\end{figure}

\begin{figure}[h!]
\centerline{\includegraphics[width=\textwidth]{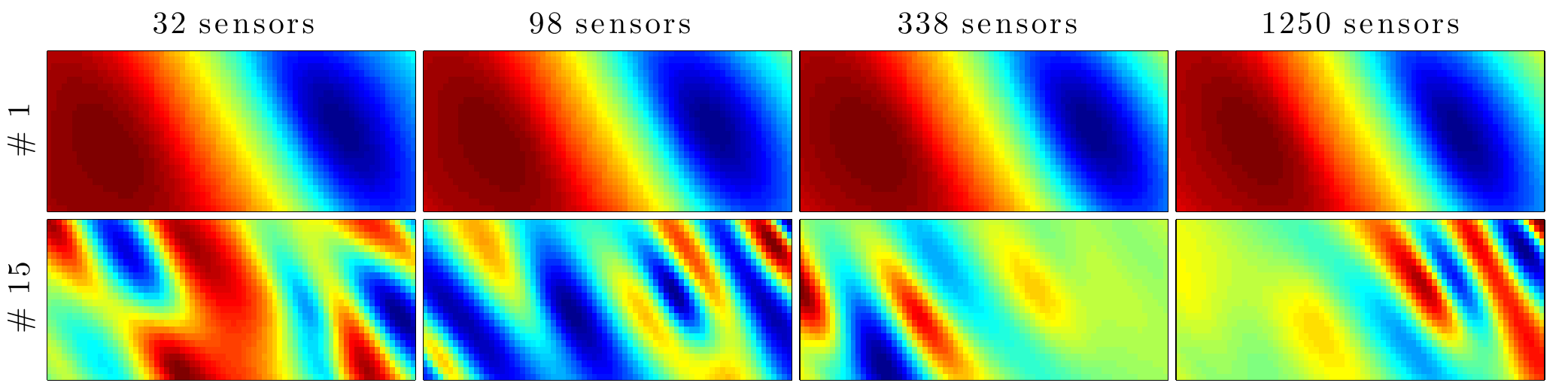}}
\caption{The first and fifteenth LIS basis vectors for each of the four data sets.}
\label{fig:data_ref_basis}
\end{figure}

\section{Example 2: atmospheric remote sensing}
\label{sec:remote}

In this section, we apply the dimension reduction approach to a realistic atmospheric satellite remote sensing problem. The problem is to invert the concentrations of various gases in the atmosphere using the measurement system applied in the GOMOS satellite instrument, which stands for \textit{Global Ozone MOnitoring System}.

GOMOS is an instrument on board ESA's Envisat satellite, and was operational for about 10 years before the connection with the satellite was lost in May 2012. The GOMOS instrument performs so-called star occultation measurements; it measures, at different wavelengths, the absorption of starlight as it travels through the atmosphere. Different gases in the atmosphere (such as ozone, nitrogen dioxide and aerosols) leave fingerprints in the measured intensity spectra. The task of the inversion algorithm is to infer the concentrations of these gases based on the measurements. %

The GOMOS inverse problem is known to be ill-posed; the intensity spectra may contain strong information about the major gases (like O$_3$) at some altitudes, whereas some minor gases (like aerosols) at some altitudes may be practically unidentifiable and totally described by the prior. Thus, the GOMOS problem is a good candidate for our approach: the dimension of the likelihood informed subspace is expected to be small and the prior contribution large. 

Next, we briefly present the GOMOS theory and the inverse problem setup. For more details about the GOMOS instrument and the Bayesian treatment of the inverse problem, see \cite{Haario_2004} and the references therein.

\subsection{The GOMOS model}

The GOMOS instrument repeatedly measures light intensities $I_\lambda$ at different wavelengths $\lambda$. First, a reference intensity spectrum $I_{\mathrm{ref}}$ is measured above the atmosphere. The so-called transmission spectrum is defined as $T_\lambda=I_\lambda / I_{\mathrm{ref}}$. The transmissions measured at wavelength $\lambda$ along the ray path $l$ are modelled using Beer's law:
\begin{equation}
T_{\lambda,l}=\exp \left( - \int_l \sum_{\mathrm{gas}} \alpha_\lambda^{\mathrm{gas}}(z(s)) \rho^{\mathrm{gas}}(z(s))ds \right),
\label{mod}
\end{equation}
where $\rho^{\mathrm{gas}}(z(s))$ is the density of a gas (unknown parameter) at tangential height $z$. The so called cross-sections $\alpha_\lambda^{\mathrm{gas}}$, known from laboratory measurements, define how much a gas absorbs light at a given wavelength. 

To approximate the integrals in (\ref{mod}), the atmosphere is discretized. The geometry used for inversion resembles an onion: the gas densities are assumed to be constant within spherical layers around the Earth. The GOMOS measurement principle is illustrated in Figure \ref{fig:gomos_setting} below. 

\begin{figure}[h!]
\centerline{\includegraphics[width=0.9\textwidth]{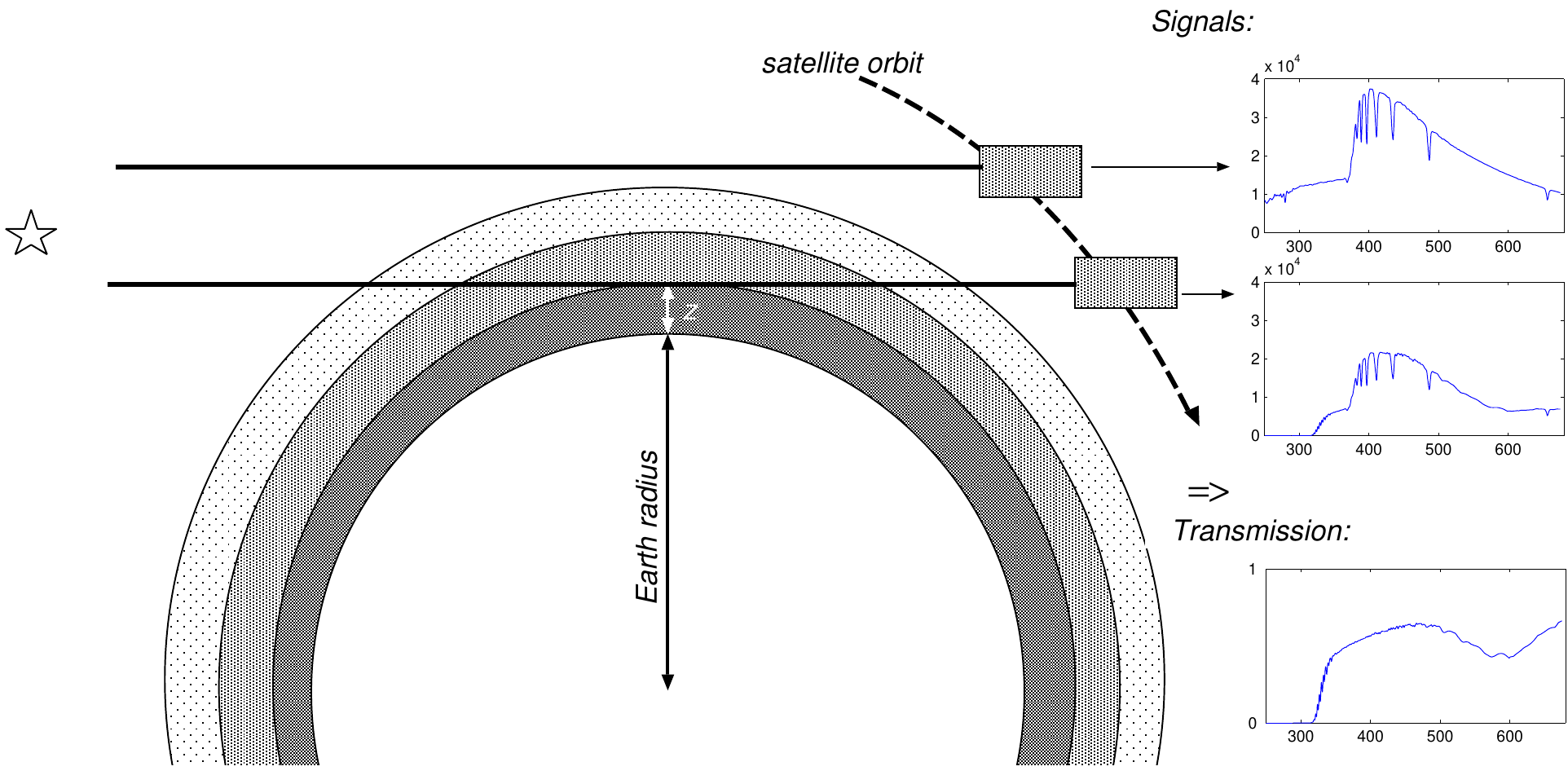}}
\caption{The principle of the GOMOS measurement. The reference intensity is measured above the atmosphere. The observed transmission spectrum is the attenuated spectrum (measured through the atmosphere) divided by the reference spectrum. The atmosphere is presented locally as spherical layers around the Earth. Note that the thickness of the layers is much larger relative to the Earth in this figure than in reality. The figure is adopted from \cite{Haario_2004}, with the permission of the authors.}
\label{fig:gomos_setting}
\end{figure}

Here, we assume that the cross-sections do not depend on height. In the inverse problem we have $n_{\mathrm{gas}}$ gases, $n_{\lambda}$ wavelengths, and the atmosphere is divided into $n_{\mathrm{alts}}$ layers. The discretisation is fixed so that number of measurement lines is equal to the number of layers. Approximating the integrals by sums in the chosen grid, and combining information from all lines and all wavelengths, we can write the model in matrix form as follows:
\begin{equation}
T=\exp(-CB^\top A^\top), 
\end{equation}
where $T \in \mathbb{R}^{n_{\lambda} \times n_{\mathrm{alts}}}$ are the modelled transmissions, $C \in \mathbb{R}^{n_{\lambda} \times n_{\mathrm{gas}}}$ contains the cross-sections, $B \in \mathbb{R}^{n_{\mathrm{alts}} \times n_{\mathrm{gas}}}$ contains the unknown densities and $A \in \mathbb{R}^{n_{\mathrm{alts}} \times n_{\mathrm{alts}}}$ is the geometry matrix that contains the lengths of the lines of sight in each layer. %

Computationally, it is convenient to deal with vectors of unknowns. We vectorize the above model using the identity $\mathrm{vec}(CB^\top A^\top)=(A \otimes C) \mathrm{vec}(B^\top)$, where $\otimes$ denotes the Kronecker product and $\mathrm{vec}$ is the standard vectorization obtained by stacking the columns of the matrix argument on top of each other. Thus, the likelihood model is written in vector form as follows:
\begin{equation}
\data = \mathrm{vec}(T)+\error = \exp \left( -(A \otimes C) \mathrm{vec}(B^\top) \right)+\error,
\label{modvec}
\end{equation}
where $\error$ is the measurement error, for which we apply an independent Gaussian model with known variances.

Note that, in principle, the model (\ref{modvec}) could be linearized by taking logarithms of both sides, which is usually done for such tomography problems (like X-ray computerized tomography). For this problem, linearisation can cause problems, since the signal from the star is often smaller compared to the background noise in the measurement. 

\subsection{Data and prior}

Here, we generate synthetic data by solving the forward model (\ref{modvec}) with known gas densities $x$. In the example, we have 4 gas profiles to be inverted. The atmosphere is discretized into 50 layers, and the total dimension of the problem is thus 200. The simulated data are illustrated in Figure \ref{fig:gomos_data}.

\begin{figure}[h!]
\centerline{\includegraphics[width=\textwidth]{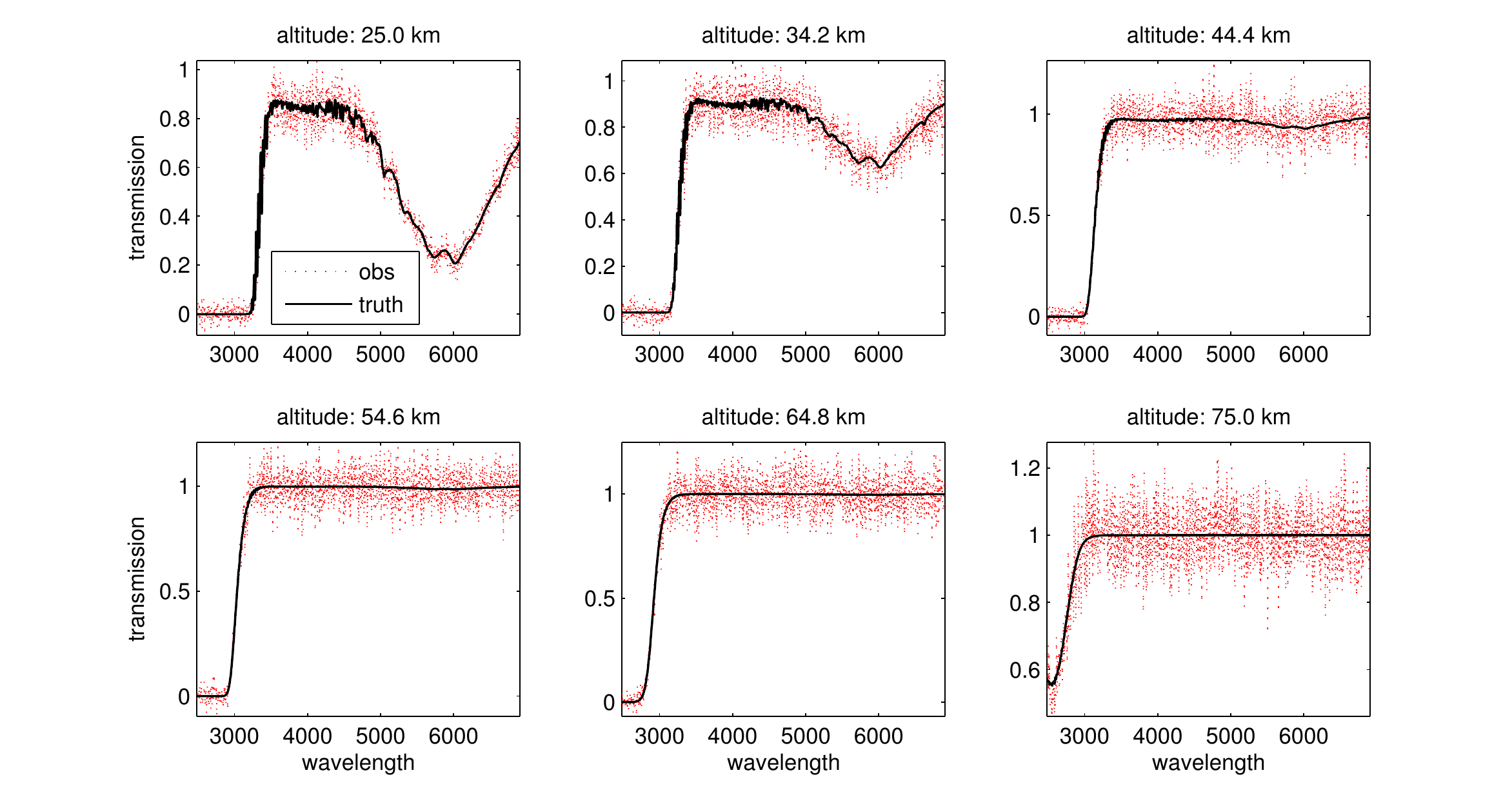}}
\caption{GOMOS example setup: the true transmissions (black) and the observed transmissions (red) for 6 altitudes.}
\label{fig:gomos_data}
\end{figure}

We estimate the log-profiles $\param=\log(\mathrm{vec}(B^\top))$ of the gases instead of the densities $B$ directly. We set a Gaussian process prior for the profiles, which yields  $\param_i \sim N(\mu_i,\Sigma_i)$, where $\param_i$ denotes the elements of vector $\param$ corresponding to gas $i$. The elements of the 50 $\times$ 50 covariance matrices are calculated based on the squared exponential covariance function 
\begin{equation}
C_i(s,s^{\prime})=\sigma_i\exp(-(s-s^{\prime})^2/2s_{0,i}^2),
\end{equation}
where the parameter values are $\sigma_1=5.22$, $\sigma_2=9.79$, $\sigma_3=23.66$, $\sigma_4=83.18$, and $s_{0,i}=10$ for all $i$. %
The priors are chosen to promote smooth profiles and to give a rough idea about the magnitude of the density values. The prior is illustrated in Figure \ref{fig:gomos_prior}. 

\begin{figure}[h!]
\centerline{\includegraphics[width=\textwidth]{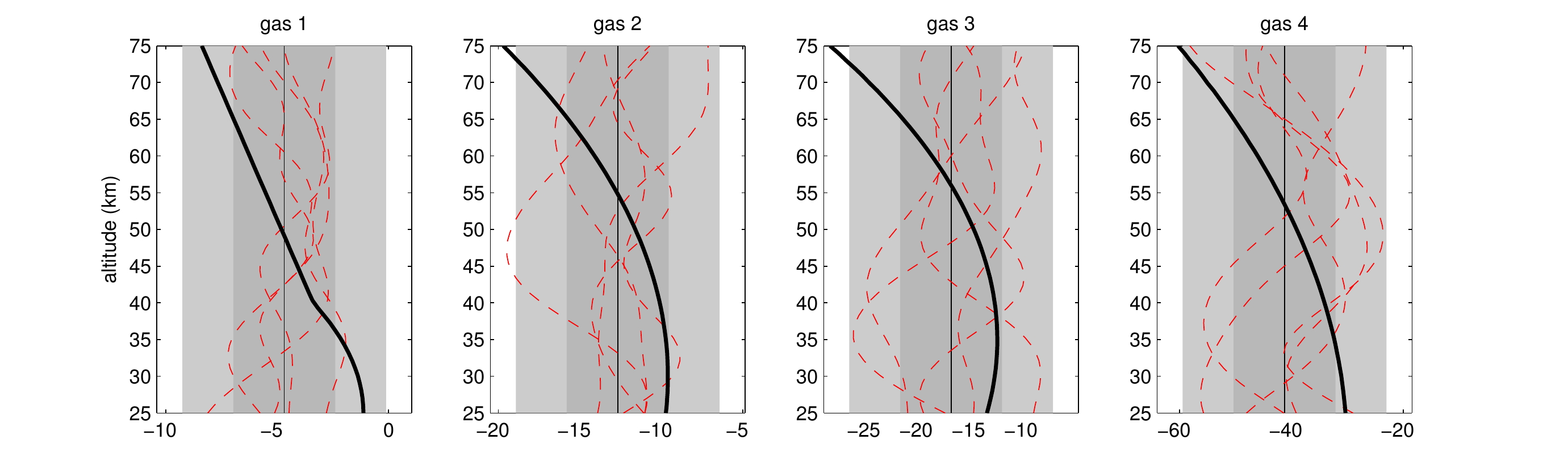}}
\caption{True log-profiles for the 4 gases (black solid lines), 50\% and 95\% confidence envelopes for the prior (grey areas) and 5 samples from the prior (red dashed lines).}
\label{fig:gomos_prior}
\end{figure}

\subsection{Inversion results}

In this particular synthetic example, we know that gas 1 is very well identified by the data. The data also contain information about gases 2 and 3 at some altitudes. Gas 4, on the other hand, is totally unidentified by the data.

The LIS is constructed using $200$ samples---i.e., 200 iterations of Algorithm~\ref{algo:subspace}---starting with the Hessian at the posterior mode. The subspace convergence diagnostic and the generalized eigenvalues are shown in Figure \ref{fig:gomos_LIS_conv}. We choose the truncation thresholds $\tau_{loc} = \tau_{g} = 0.5$. The dimension of the LIS in the end was 22.

\begin{figure}[h!]
\centerline{\includegraphics[width=\textwidth]{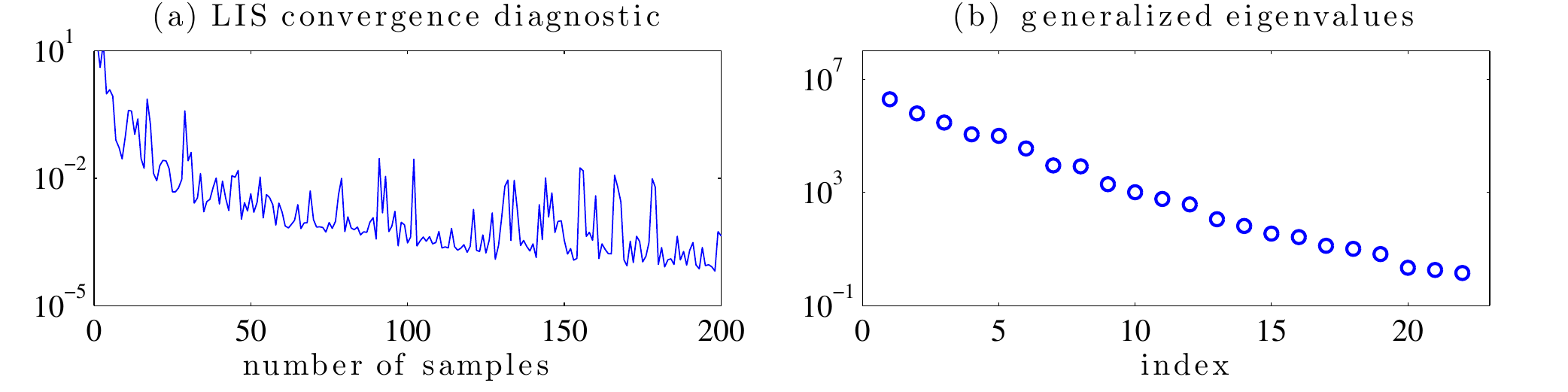}}
\caption{Left: the convergence diagnostic \eqref{eq:lis_conv} versus the number of samples used to construct the LIS. Right: the generalized eigenvalues associated with the final LIS.}
\label{fig:gomos_LIS_conv}
\end{figure}

We compute $10^6$ samples in both the LIS and in the full 200-dimensional space using the Hessian-preconditioned MALA algorithm. In Figure \ref{fig:gomos_est}, the first two columns show the mean gas profile and the mean $\pm$ 1 and 2 standard deviations in the LIS and in the complement space (CS). The third column shows the combined posterior from the LIS and the CS; for comparison, results from full-space MCMC are shown in the fourth column. Note the different scales on the horizontal axes throughout the figure. We observe that the subspace approach, where MCMC is applied only in a 22-dimensional space, yields results very similar to those of full MCMC. In addition, comparing the contributions of the LIS and CS indicates that gas 1 is dominated by the likelihood, whereas the posterior distribution of gas 4 is entirely determined by the prior. Note that the CS contribution for gas 1 is tiny (check the scale), while the LIS contribution for gas 4 is also very small. For gases 2 and 3, the lower altitudes are likelihood-dominated, while the higher altitudes have more contribution from the prior. The full-space MCMC results for gas 4 show a slightly non-uniform mean, but this appears to be the result of sampling variance. By avoiding sampling altogether in the CS, the subspace approach most likely yields a more accurate posterior in this case.

\begin{figure}[h!]
\centerline{\includegraphics[width=\textwidth]{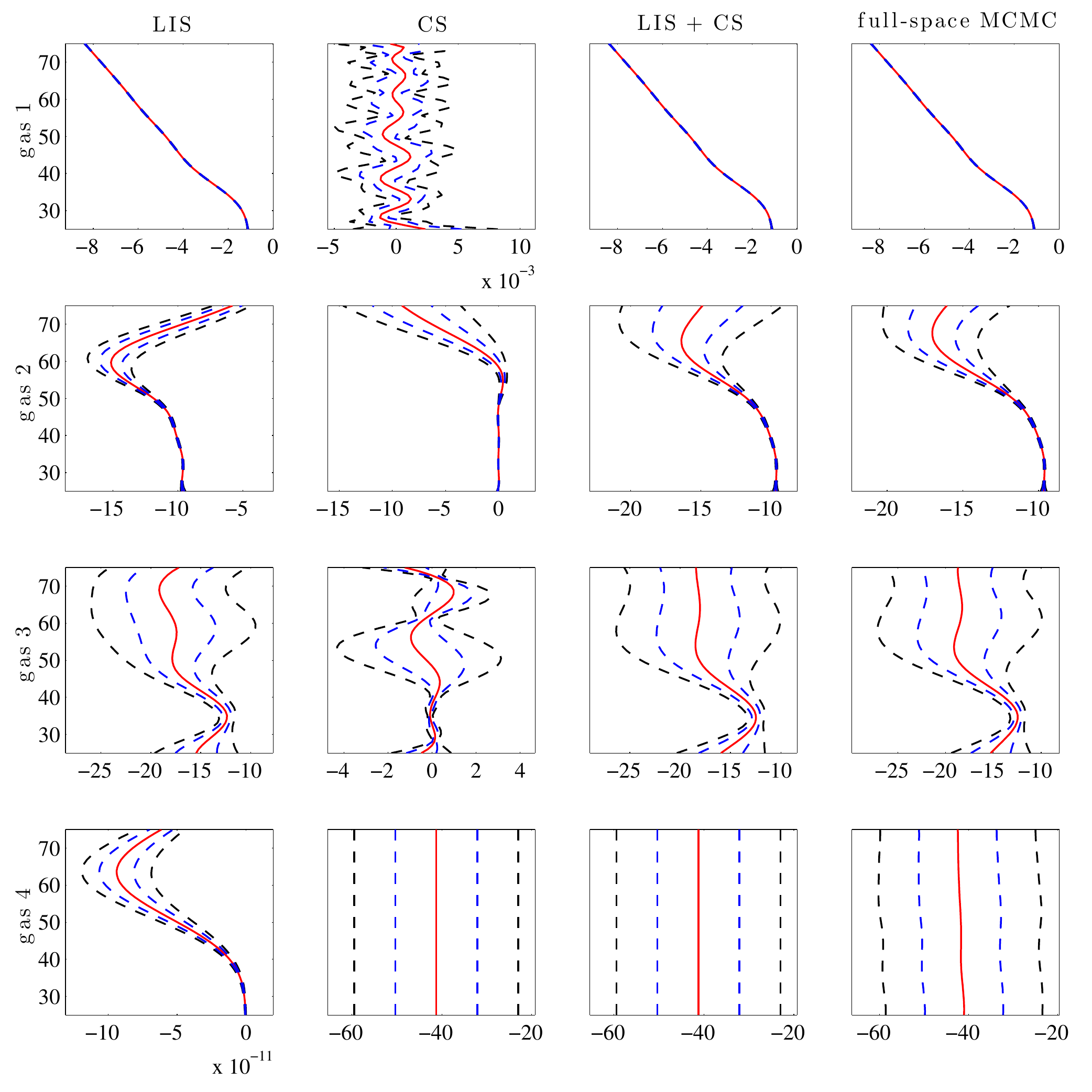}}
\caption{Mean and $\pm$1/$\pm$2 standard deviations for the 4 gas profiles computed from the LIS samples alone (1st column), CS alone (2nd column) and when the LIS and CS are combined (3rd column). The same quantities computed from full-space MCMC are given in the 4th column.}
\label{fig:gomos_est}
\end{figure}

To further illustrate the approach, we plot the first six basis vectors of the LIS in Figure \ref{fig:gomos_LIS}. One can see that the first basis vectors mainly include features of gas 1, which is most informed by the data. The first basis vectors also contain some features of gases 2 and 3 in lower altitudes. Gas 4 is not included in the LIS at all. 

\begin{figure}[h!]
\centerline{\includegraphics[width=\textwidth]{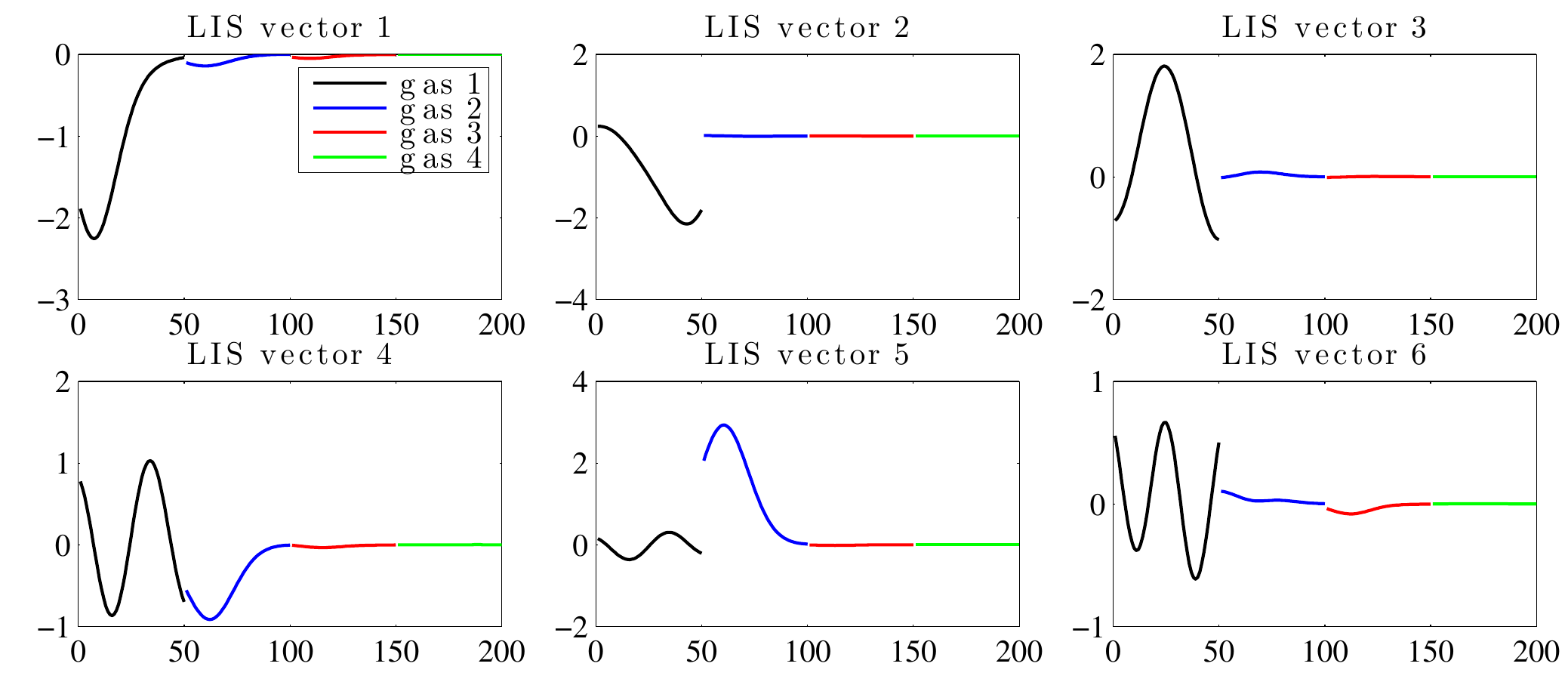}}
\caption{The first six LIS basis vectors for the remote sensing example. The colors indicate the components of the unknown vector corresponding to the different gases. In each subfigure, the $x$-axis denotes the index of the parameter vector, and, for each gas, the components are ordered from low altitudes to high altitudes. (For example, the black line in each figure shows gas 1 profiles from low altitudes to high altitudes, etc.)}
\label{fig:gomos_LIS}
\end{figure}

The dimension reduction obtained via the subspace approach is expected to yield better mixing than the full-space MCMC. For the GOMOS case, the chain autocorrelations for subspace and full-space MCMC are compared in Figure \ref{fig:gomos_auto}. The subspace sampler shows much faster decay of the autocorrelations than full-space MCMC. %

In this test case, the subspace MCMC also has lower computational cost compared to full-space MCMC. 
To simulate a Markov chain for $10^6$ iterations, the subspace MCMC consumed about $2560$ seconds of CPU time, while the full-space MCMC cost $3160$ CPU seconds. 
We note that the CPU time reduction is not as significant as the elliptic example, because the prior covariance is a $200\times200$ dimensional matrix, which is much smaller than the covariance matrix used in the elliptic example. 
To construct the LIS, we simulated Algorithm \ref{algo:subspace} for $200$ iterations. This cost about $136$ seconds of CPU time, which is only about $4.3\%$ of the CPU time used to run full-space MCMC for $10^6$ steps.

\begin{figure}[h!]
\centerline{\includegraphics[width=\textwidth]{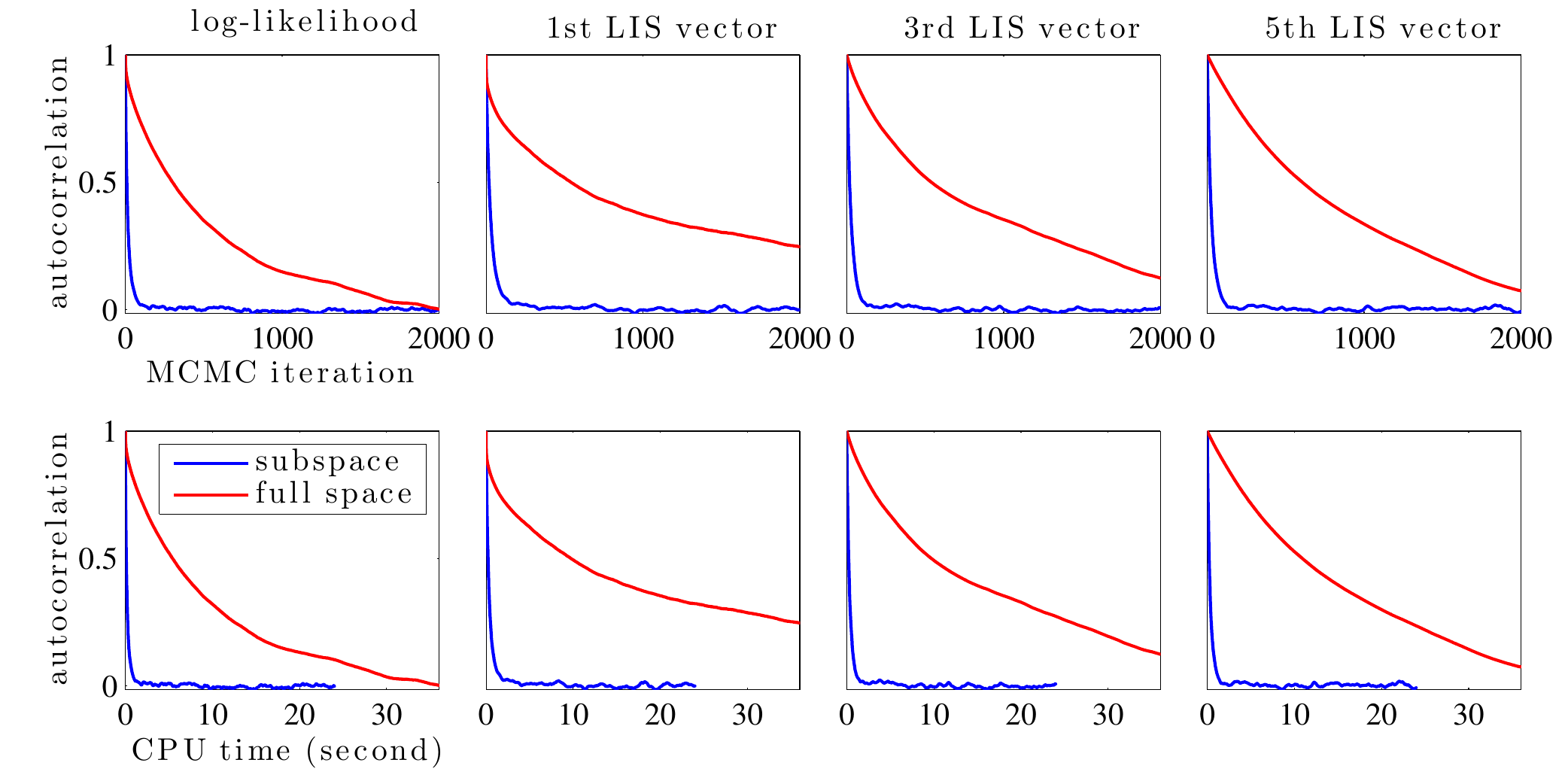}}
\caption{Autocorrelations of full-space (red) and subspace (blue) MCMC for the log-likelihood (1st column) and for the samples projected onto the first, third, and fifth LIS basis vectors (2nd, 3rd and 4rd columns). Top rows shows the autocorrelations computed per MCMC step and bottom row per CPU time.}
\label{fig:gomos_auto}
\end{figure}

\section{Conclusions}
\label{sec:conclusions}

In this paper, we present a new approach for dimension reduction in nonlinear inverse problems with Gaussian priors. Our approach is based on dividing the parameter space into two subspaces: a likelihood-informed subspace (LIS) where the likelihood has a much greater influence on the posterior than the prior distribution, and the complement to the LIS where the Gaussian prior dominates. 
We explore the posterior projected onto the LIS (the ``difficult'' and non-Gaussian part of the problem) with Markov chain Monte Carlo while treating the complement space as exactly Gaussian. This approximation allows us to analytically integrate many functions over the complement space when estimating their posterior expectations; the result is a Rao-Blackwellization or de-randomization procedure that can greatly reduce the variance of posterior estimates. Particularly in inverse problems---where information in the data is often limited and the solution of the problem relies heavily on priors---the dimension of the LIS is expected to be small, and the majority of the directions in the parameter space can be handled analytically.

The dimension reduction approach is based on theory developed for the linear case; in \cite{Linear_Redu_2014} it is shown that in linear-Gaussian problems, the eigendecomposition of the prior-preconditioned log-likelihood Hessian yields an optimal low-rank update from the prior to the posterior, which can be interpreted in terms of a projector whose range is the LIS. Here, we generalize the approach to nonlinear problems, where the log-likelihood Hessian varies over the parameter space. Our solution is to construct many local likelihood-informed subspaces over the support of the posterior and to combine them into a single global LIS. We show how the global LIS can be constructed efficiently in an adaptive manner, starting with the LIS computed at the posterior mode and iteratively enriching the global LIS until a weighted subspace convergence criterion is met. 

We demonstrate the approach with two numerical examples. First is an elliptic PDE inverse problem, based on a simple model of subsurface flow. Though the dimension of the parameter space in our experiments ranges from 1200 to 10800, the dimension of the LIS remains only around 20 and is empirically discretization-invariant. Exploring the LIS by MCMC and analytically treating the Gaussian complement produces mean and variance fields very similar to those computed via MCMC in the full space. Yet the mixing properties and the computational cost of MCMC in the LIS are dramatically improved over those of full-space MCMC. Our second demonstration is an atmospheric remote sensing problem, where the goal is to infer the concentrations of chemical species in the atmosphere using star occultation measurements, as on the satellite-borne GOMOS instrument. The dimension of the full problem used here was 200 (four gaseous species and 50 altitudes for each), while the dimension of the LIS was 22. Again, dimension reduction significantly improves the mixing properties of MCMC without sacrificing accuracy.

To conclude, our dimension reduction approach appears to offer an efficient way to probe and exploit the structure of nonlinear inverse problems in order to perform Bayesian inference at a large scale, where standard algorithms are plagued by the curse of dimensionality. The approach also opens up interesting further research questions: it may be useful, for instance, to apply reduced-order and surrogate modeling techniques in the LIS, making them applicable to much larger problems than before.

\ack We acknowledge Marko Laine and Johanna Tamminen from the Finnish
Meteorological Institute for providing us with the GOMOS figure and
codes that served as the baseline for our implementation for the
remote sensing example. This work was supported by the US Department
of Energy, Office of Advanced Scientific Computing (ASCR), under grant
numbers DE-SC0003908 and DE-SC0009297.


\begin{thebibliography}{10}
\providecommand{\url}[1]{\texttt{#1}}
\providecommand{\urlprefix}{URL }
\expandafter\ifx\csname urlstyle\endcsname\relax
  \providecommand{\doi}[1]{doi:\discretionary{}{}{}#1}\else
  \providecommand{\doi}{doi:\discretionary{}{}{}\begingroup
  \urlstyle{rm}\Url}\fi
  
\bibitem{Tarantola_2004}
A.~Tarantola.
\newblock {\em Inverse Problem Theory and Methods for Model Parameter
  Estimation}.
\newblock Society for Industrial Mathematics, Philadelphia, 2005.

\bibitem{Kaipio_2005}
J.~P. Kaipio and E.~Somersalo.
\newblock {\em Statistical and Computational Inverse Problems}, volume 160.
\newblock Springer, New York, 2004.

\bibitem{Stuart_2010}
A.~M. Stuart.
\newblock Inverse problems: a {B}ayesian perspective.
\newblock {\em Acta Numerica}, 19:451--559, 2010.

\bibitem{MCMC_practice}
W.~R. Gilks, S.~Richardson, and D.~J. Spiegelhalter, editors.
\newblock {\em {M}arkov Chain {M}onte {C}arlo in practice}, volume~2.
\newblock CRC press, 1996.

\bibitem{Liu_2001}
J.~S. Liu.
\newblock {\em {M}onte {C}arlo strategies in Scientific Computing}.
\newblock Springer, New York, 2001.

\bibitem{Handbook_MCMC}
S.~Brooks, A.~Gelman, G.~Jones, and X.~L. Meng, editors.
\newblock {\em Handbook of {M}arkov Chain {M}onte {C}arlo}.
\newblock Taylor \& Francis, 2011.

\bibitem{RGG_1997}
G.~O. Roberts, A.~Gelman, and W.~R. Gilks.
\newblock Weak convergence and optimal scaling of random walk {M}etropolis
  algorithms.
\newblock {\em Annals of Applied Probability}, 7:110--120, 1997.

\bibitem{Roberts_1998}
G.~O. Roberts and J.~S. Rosenthal.
\newblock Optimal scaling of discrete approximations to {L}angevin diffusions.
\newblock {\em Journal of the Royal Statistical Society: Series B (Statistical
  Methodology)}, 60:255--268, 1998.

\bibitem{Roberts_2001}
G.~O. Roberts and J.~S. Rosenthal.
\newblock Optimal scaling of various {M}etropolis-{H}astings algorithms.
\newblock {\em Statistical Science}, 16(4):351--367, 2001.

\bibitem{Mattingly_2012}
J.~C. Mattingly, N.~Pillai, and A.~M. Stuart.
\newblock Diffusion limits of the random walk {M}etropolis algorithm in high
  dimensions.
\newblock {\em Annals of Applied Probability}, 22:881--930, 2012.

\bibitem{PST_2012}
N.~S. Pillai, A.~M. Stuart, and A.~H. Thiery.
\newblock { Optimal scaling and diffusion limits for the {L}angevin algorithm
  in high dimensions},.
\newblock {\em Annals of Applied Probability}, 22:2320--2356, 2012.

\bibitem{BRSV_2008}
A.~Beskos, G.~O. Roberts, A.~M. Stuart, and J.~Voss.
\newblock {MCMC} methods for diffusion bridges.
\newblock {\em Stochastic Dynamics}, 8(3):319--350, 2008.

\bibitem{CRSW_2012}
S.~L. Cotter, G.~O. Roberts, A.~M. Stuart, and D.~White.
\newblock {MCMC} methods for functions: modifying old algorithms to make them
  faster.
\newblock {\em Statistical Science}, 28:424--446, 2013.

\bibitem{Flath_etal_2011}
H.~P. Flath, L.~C. Wilcox, V.~Akcelik, J.~Hill, B.~van Bloemen~Waanders, and
  O.~Ghattas.
\newblock Fast algorithms for {B}ayesian uncertainty quantification in
  large-scale linear inverse problems based on low-rank partial hessian
  approximations.
\newblock {\em SIAM Journal on Scientific Computing}, 33(1):407--432, 2011.

\bibitem{Linear_Redu_2014}
A.~Spantini, A.~Solonen, T.~Cui, J.~Martin, L.~Tenorio, and Y.~Marzouk.
\newblock Optimal low-rank approximation of linear {B}ayesian inverse problems.
\newblock {\em  arXiv preprint}, arXiv:1407.3463, 2014.

\bibitem{Marzouk_2009}
Y.~M. Marzouk and H.~N. Najm.
\newblock Dimensionality reduction and polynomial chaos acceleration of
  {B}ayesian inference in inverse problems.
\newblock {\em Journal of Computational Physics}, 228:1862--1902, 2009.

\bibitem{Karhunen_1947}
K.~Karhunen.
\newblock \"{U}ber lineare methoden in der wahrscheinlichkeitsrechnung.
\newblock {\em Ann. Acad. Sci. Fennicae. Ser. A. I. Math.-Phys}, 37:1--79,
  1947.

\bibitem{Loeve_1978}
M.~Lo\`{e}ve.
\newblock {\em Probability theory, Vol. II}, volume~46 of {\em Graduate Texts
  in Mathematics}.
\newblock Springer-Verlag, Berlin, 4 edition, 1978.

\bibitem{Martin_2012}
J.~Martin, L.~C. Wilcox, C.~Burstedde, and O.~Ghattas.
\newblock A stochastic {N}ewton {MCMC} method for large-scale statistical
  inverse problems with application to seismic inversion.
\newblock {\em SIAM Journal on Scientific Computing}, 34(3):A1460--A1487, 2012.

\bibitem{Petra_stochnewton2013}
Noemi Petra, James Martin, Georg Stadler, and Omar Ghattas.
\newblock A computational framework for infinite-dimensional {B}ayesian inverse
  problems: Part {II.} stochastic {N}ewton {MCMC} with application to ice sheet
  flow inverse problems.
\newblock {\em SIAM Journal on Scientific Computing}, to appear, 2014.
\newblock arXiv:1308.6221.

\bibitem{LSS_2009}
M.~Lassas, E.~Saksman, and S.~Siltanen.
\newblock Discretization invariant {B}ayesian inversion and {B}esov space
  priors.
\newblock {\em Inverse Problems and Imaging}, 3(1):87--122, 2009.

\bibitem{FB_1999}
W.~F\"{o}rstner and M.~Boudewijn.
\newblock A metric for covariance matrices.
\newblock In {\em Geodesy-The Challenge of the 3rd Millennium}, pages 299--309. Springer Berlin Heidelberg, 2003.

\bibitem{Golub_2012}
G.~H. Golub and C.~F.~Van Loan.
\newblock {\em Matrix Computations}.
\newblock JHU Press, 2012.

\bibitem{SVD_HMT_2011}
N.~Halko, P.~Martinsson, and J.~A. Tropp.
\newblock Finding structure with randomness: Probabilistic algorithms for
  constructing approximate matrix decompositions.
\newblock {\em SIAM Review}, 53(2):217--288, 2011.

\bibitem{SVD_Liberty_etal_2007}
E.~Liberty, F.~Woolfe, P.~G. Martinsson, V.~Rokhlin, and M.~Tygert.
\newblock Randomized algorithms for the low-rank approximation of matrices.
\newblock {\em Proceedings of the National Academy of Sciences},
  104(51):20167--20172, 2007.

\bibitem{Haario_2001}
H.~Haario, E.~Saksman, and J.~Tamminen.
\newblock An adaptive {M}etropolis algorithm.
\newblock {\em Bernoulli}, 7(2):223--242, 2001.

\bibitem{AM_adapt_2006}
C.~Andrieu and E.~Moulines.
\newblock On the ergodicity properties of some adaptive {MCMC} algorithms.
\newblock {\em The Annals of Applied Probability}, 16(3):1462--1505, 2006.

\bibitem{Atchade_2006}
Y.~F. Atchade.
\newblock An adaptive version for the {M}etropolis adjusted {L}angevin
  algorithm with a truncated drift.
\newblock {\em Methodology and Computing in {A}pplied Probability},
  8(2):235--254, 2006.

\bibitem{Roberts_2007}
G.~O. Roberts and J.~S. Rosenthal.
\newblock Coupling and ergodicity of adaptive {M}arkov chain {M}onte {C}arlo
  algorithms.
\newblock {\em Journal of Applied Probability}, 44(2):458--475, 2007.

\bibitem{Roberts_2009}
G.~O. Roberts and J.~S. Rosenthal.
\newblock Examples of adaptive {MCMC}.
\newblock {\em Journal of Computational and Graphical Statistics},
  18(2):349--367, 2009.

\bibitem{Girolami_2011}
M.~Girolami and B.~Calderhead.
\newblock Riemann manifold {L}angevin and {H}amiltonian {M}onte {C}arlo
  methods.
\newblock {\em Journal of the Royal Statistical Society: Series B (Statistical
  Methodology)}, 73(2):123--214, 2011.

\bibitem{CR_1996}
G.~Casella and C.~P. Robert.
\newblock Rao-{B}lackwellisation of sampling schemes.
\newblock {\em Biometrika}, 83(1):81--94, 1996.

\bibitem{Distance_Li_2009}
F.~Li, Q.~Dai, W.~Xu, and G.~Er.
\newblock Weighted subspace distance and its applications to object recognition
  and retrieval with image sets.
\newblock {\em IEEE Signal Processing Letters}, 16(3):227--230, 2009.

\bibitem{CLM_2014}
T.~Cui, K.~J.~H. Law, and Y.~M. Marzouk.
\newblock Dimension--independent likelihood--informed {MCMC}.
\newblock {\em arXiv preprint}, 2014.

\bibitem{Haario_2004}
H.~Haario, M.~Laine, M.~Lehtinen, E.~Saksman, and J.~Tamminen.
\newblock {M}arkov chain {M}onte {C}arlo methods for high dimensional inversion
  in remote sensing.
\newblock {\em Journal of the Royal Statistical Society: Series B (Statistical
  Methodology)}, 66:591--608, 2004.

\end{thebibliography}


\section*{References}
\end{document}